%% file: main.tex
\documentclass[journal,onecolumn]{IEEEtran}
\input{header.tex}
\usepackage{newtxtext}

\begin{document}
\title{Break-Resilient Codes}

\author{Canran Wang,
\IEEEmembership{Member, IEEE}, Jin Sima, \IEEEmembership{Member, IEEE} and
Netanel Raviv,
\IEEEmembership{Senior Member, IEEE}
\thanks{
This work was supported in part by NSF under
Grant CNS 2223032.
Canran Wang and Netanel Raviv are with the Department of Computer Science and Engineering,
Washington University in St. Louis, St. Louis, MO 63130 USA (email: canran@wustl.edu, netanel.raviv@wustl.edu).
Jin Sima is with the Department of Electrical and Computer Engineering, University of Illinois Urbana–Champaign, Urbana, IL 61801 USA (e-mail:
jsima@illinois.edu).
Parts of this work have previously appeared in [DOI:10.1109/ISIT57864.2024.10619135].
}}

\maketitle
\begin{abstract}

We investigate the problem of encoding data into an $(n, t)$-break-resilient code ($(n, t)$-BRC), i.e., a collections of sequences of length~$n$ from which the original data can be reconstructed even if they are adversarially broken at up to~$t$ arbitrary positions.
We establish lower bounds on the redundancy of any $(n, t)$-BRC and present code constructions that meet these bounds up to asymptotically negligible terms.
Interestingly, this problem shares similarities with the recently studied torn paper channel, which has emerged in the context of DNA data storage.

\end{abstract}
\section{Introduction}
Modern data embedding techniques increasingly operate outside traditional digital channels, in environments where data can be deliberately broken apart rather than merely corrupted by random noise.
As a motivating application, consider a digital fingerprint embedded within a 3D-printed component for authentication or traceability.
While various techniques for embedding bits in 3D prints have been proposed in the literature~\cite{elsayed2021information,voris2017three,wei2018embedding,chen2019embedding,delmotte2019blind,suzuki2017embedding,li2018printracker,harrison2012acoustic}, the data reconstruction relies on the integrity of the physical geometry of the object.
If the object is broken into several pieces, the information stored in its structure will likewise split into multiple fragments with no obvious order, and hence fails the data reconstruction process.

As such, a break-resilient fingerprinting scheme requires the decoder to recover the embedded information with access to only a jumbled collection of fragments of the embedded information.
The task is made even harder if the break positions do not follow any probabilistic distribution, but are deliberately chosen by an adversary who is fully aware of the coding scheme.
This paper focuses on such an adversarial model in which the information is broken apart by the adversary at~$t$ arbitrary positions, and the goal of the decoder is to recover the original information bits from the (at most)~$t+1$ resulting fragments in all possible cases.
Notably, the knowledgeable adversary is only constrained by the security parameter~$t$, and wishes to interfere with the decoding process as much as possible within his capability.

\subsection{Related Works}

Similar coding problems have been recently studied in the literature, motivated by the nascent technology of information storage in DNA molecules.
Since information storage in DNA molecules is restricted to short and unordered sequences, several works studied the so-called \emph{sliced-channel} model, in which the information bits are sliced at several evenly-spaced positions, producing a set of substrings of equal size~\cite{sima2021coding,sima2024robust,lenz2019coding}.
The \emph{torn paper coding} problem has been studied by~\cite{shomorony2021torn,shomorony2020communicating,ravi2021capacity}, where the information string is being cut by a probabilistic process (as opposed to the adversarial model considered herein), producing substrings of random lengths.
More closely related, the adversarial counterpart of torn-paper coding has been studied by~\cite{bar2023adversarial} with the restriction that all fragment lengths are between some upper and lower bounds.

Another related problem in the fields of computational biology and text processing was studied under the name of Minimum Common String Partition (MCSP)~\cite{goldstein2004minimum,jiang2012minimum,damaschke2008minimum,chrobak2004greedy}.
For two strings over the same alphabet, a~\emph{common string partition} is a multiset of substrings which can be obtained by partitioning either one of the strings. 
The objective of the MCSP problem is devising an algorithm which finds such a multiset of minimum size for any two strings given as input.
Observe that the minimum common string partition in our problem setting must be greater than~$t$ for every pair of codewords.
Otherwise, the adversary may confuse the decoder by deliberately breaking a codeword into~$(t+1)$ fragments which can be rearranged to obtain another codeword.
However, previous works focused on finding the MCSP for arbitrary strings efficiently, whereas we are interested in explicit code constructions.

\subsection{Our Contributions}
We extend the above line of works and provide a nearly-optimal binary code construction for a wide range of~$t$ values relative to the code length~$n$.
We refer to these codes as~$(n,t)$-\emph{break-resilient code} ($(n, t)$-BRC), and establish matching lower bounds.
Specifically, our analysis begins with a Gilbert–Varshamov type argument, which demonstrates the existence of an~$(n,t)$-BRC with redundancy~$\Omega(t\log n)$.
We then provide a simple reduction from break-resilient codes to traditional error-correcting codes, showing that~$\Omega(t\log n)$ is in fact the minimum redundancy of a binary $(n, t)$-BRC.
Finally, we give a novel construction achieving~$O(t\log n\log\log n)$ redundancy, which is optimal up to a small factor of~$\log\log n$.
 
In a nutshell, our code construction relies on identifying and utilizing short patterns in the information word, which we call \emph{beacons}.
For every pair of neighboring beacons, we record their relative order, protect these records with a systematic Reed–Solomon code, and concatenate the resulting parity symbols to the information word.
During decoding, the receiver scans the~$t+1$ fragments for adjacent beacons that remain intact, and then uses the parity symbols to reconstruct the full ordering of all beacon pairs.
This recovered ordering uniquely determines the correct sequence in which the fragments must be reassembled.
This process yields a code with redundancy~$O(t\log n\log\log n)$.

As a bonus, the proposed~$(n, t)$-BRC can tolerate the loss of any of the~$t+1$ fragments whose length is~$O(\log n)$.
In other words, the decoding is guaranteed to succeed even in the case where the adversary is allowed to hide small fragments; this property makes our~$(n, t)$-BRC even more practical in its potential real-world applications.
We also emphasize that the problem is more challenging over small alphabets; over (very) large alphabets a simple histogram-based construction achieves a constant number of redundant symbols (Appendix~\ref{appendix:histogram}).

The rest of this paper is organized as follows.
Section~\ref{section:preliminaries} provides a formal definition of the problem and clarifies notations.
Section~\ref{section:bounds} discusses bounds on the redundancy of an~$(n, t)$-BRC.
Section~\ref{section:code} details the construction of a nearly optimal~$(n, t)$-BRC, as well as its decoding algorithm.
Section~\ref{section:redundancy} analyzes the redundancy of such code.
Finally, Section~\ref{section:discussion} discusses different aspects of the code and suggests potential future research directions.

\section{Problem Definition and Preliminaries}\label{section:preliminaries}
Our setup includes an encoder which holds a string~$\bfx\in\bi^k$ (an information word) for some integer~$k$, that is encoded to a string~$\bfc\in\bi^n$ (a codeword) for some integer~$n>k$.
For a security parameter~$t$, an adversary breaks these~$n$ bits at arbitrary~$t$ positions or less, resulting in a multiset of at most~$t+1$ fragments. 
\begin{example}
    For~$\bfc=0100011100$ and~$t=3$, the possible fragment multisets include 
\begin{align*}
    \{\{ 0,0,1000,1110 \}\}, \{\{ 010,00111,00 \}\}\text{, and }\{\{ 01000,11100 \}\}.
\end{align*}    
\end{example}

These fragments are given to the decoder in an unordered fashion, and the goal of the decoder is to reconstruct~$\bfx$ exactly in all cases. 
The associated set of codewords in~$\bi^n$ is called a~\emph{$(n,t)$-break-resilient code}, and is denoted by~$\cC$.

The figure of merit of a given code is its \emph{redundancy}, i.e., the quantity~$n-\log|\cC|$, where~$|\cdot|$ denotes size.
Although our context implies that~$t$ is a small number relative to~$n$, we choose not refer to it as a constant in our asymptotic analysis in order to better understand the fine dependence of our scheme on it; in essence, our scheme applies to any~$n$ and any $t=o(\frac{n}{\log n\log\log n})$.

Further, we also assume that the fragments are \emph{oriented}.
That is, for any given fragment~$\bff=(c_{i},\ldots,c_{i+r})$ taken from a codeword~$\bfc=(c_1,\ldots, c_n)$, the decoder does not know the correct values of the indices~$i,\ldots,i+r$, but does know that the bit~$c_i$ is positioned to the left of the bit~$c_{i+r}$.
Orientation of fragments can be achieved rather easily by simple engineering solutions~\cite{wang2025secure}, thus we sidestep this issue in order to elucidate the problem more clearly.

Throughout this paper we use standard notations for string manipulation, such as~$\circ$ to denote concatenation, $|\bfx|$ to denote length, and for a string~$\bfx=(x_1,\ldots,x_n)$ and~$1\le a<b\le n$ we let~$\bfx[a:b]=(x_a,x_{a+1},\ldots,x_b)$, and~$\bfx[a:]=(x_a,\ldots,x_{|\bfx|})$, as well as~$[n]\triangleq\{1,2,\ldots,n\}$ for a positive integer~$n$.

\begin{remark}
    While the present work abstracts away the specifics of the motivating applications and the underlying bit-embedding techniques, in concurrent research we have developed a break-resilient coding framework tailored for forensic fingerprinting.
    Our framework embeds bits in 3D-printed objects by modulating layer widths and has been experimentally validated on a commodity printer using standard 3D printing software.
    The framework follows similar adjacency matrix-based ideas, and yet diverges from from what described next due to engineering and complexity constraints which are beyond the scope of this paper.
    For further details, we refer the interested readers to~\cite{wang2025secure}.
\end{remark}

\begin{remark}
We argue that the adversarial torn paper model of~\cite{bar2023adversarial}, where all fragment lengths lie between some~$L_\text{min}$ and~$L_\text{max}$, is suboptimal for our motivating application.
In contrast, our model imposes no lower bound on the number of bits contained in a fragment, which potentially could be as few as a single bit.
It instead limits~\emph{only the number of breaks} that may be inflicted on the object.
This limit is naturally dictated by practical factors such as available time, access to tools, or the physical strength of the adversary.
\end{remark}

Finally, we make use of the following two existing notions from coding theory and data structures. 
\subsection{Mutually uncorrelated codes}\label{section:MU}
A \emph{mutually uncorrelated} (MU) code is a set of codewords such that the prefix of one does not coincide with the suffix of any (potentially identical) other.
MU codes were introduced and investigated for synchronization purposes~\cite{levenshtein1970maximum,gilbert1960synchronization}. 
Later, constructions, bounds, and applications of MU codes have been extensively studied under various names such as~\emph{cross-bifix-free codes}~\cite{bajic2004distributed,bajic2014simple,chee2013cross,bilotta2012new} and~\emph{non-overlapping codes}~\cite{blackburn2015non,wang2022q}.
Recently, MU codes have been applied to DNA-based data storage architectures~\cite{tabatabaei2015rewritable, yazdi2017portable, levy2018mutually}.
\begin{example}
    An MU code with~$n_\text{MU}=15$ in which every two different codewords are mutually uncorrelated:
    \begin{equation*}
        \{ 000011001000101, 000010100010011, 00001 1101011011,000011101011001 \}
    \end{equation*}
\end{example}

Notably, Levy~\emph{et al.}~\cite{levy2018mutually} provides a simple yet efficient construction of binary MU code.
In a nutshell, it firstly maps the information word~$\bfx\in\bi^{k_\text{MU}}$ to a binary sequence~$\bfy\in\bi^{k_\text{MU}+1}$ that is free of zero runs (i.e., all-$0$ substrings) longer than~$\ceil{\log (k_\text{MU})}$; this process introduces~$1$ redundant bit.
The corresponding MU codeword is then defined as
\begin{equation*}
    0^{\ceil{\log (k_\text{MU})}+1}\circ1\circ\bfy\circ1,
\end{equation*}
i.e., a binary string of~$(\ceil{\log (k_\text{MU})}+1)$~$0$'s, followed by~$\bfy$ surrounded by two~$1$'s.
Obviously, two codewords of this form cannot overlap with each other, and this method introduces~$\ceil{\log(k_\text{MU})}+4$ redundant bits in total, which is no more than~$\ceil{\log (n_\text{MU})}+4$, where~$n_\text{MU}=k_{\text{MU}}+\ceil{\log (k_{\text{MU}})}+4$ is the code's length.

\subsection{Key-Value stores}
Also known as a~\emph{map} and a~\emph{dictionary}, a key-value (KV) store is a fundamental data structure that organizes data as a collection of key-value pairs and has been widely used in computer programming.
In a KV store, a~\emph{key} is a unique identifier used to retrieve the associated~\emph{value}.
Specifically, the operation~$\text{KV}(\texttt{key})$ returns~$\texttt{value}$ if the pair~$(\texttt{key},\texttt{value})$ is stored in the KV store, and otherwise returns some designated symbol that indicates failure.
For ease of demonstration, we employ KV stores in the description of our proposed algorithms.

\section{Bounds}\label{section:bounds}
In this section we establish lower bounds on the redundancy of an~$(n, t)$-BRC~$\cC$.
First, we show that there exists an $(n,t)$-BRC with redundancy $O(t\log n)$, and then we prove that no code can outperform this bound.
We begin with the notion of \emph{$t$-confusability}, which underpins the results that follow.

\begin{definition}[$t$-confusability]
    Two words~$\bfx,\bfy\in\bi^n$ are~\emph{$t$-confusable} if there exists at most~$t$ break positions in~$\bfx$ and at most~$t$ break positions in~$\bfy$ which produce an identical multiset of at most~$t+1$ fragments.
\end{definition}

\begin{example}
    The words~$\bfc=0100011100$ and $\bfc'=1000100011$ are~$2$-confusable, since the break patterns
    \begin{equation*}
\bfc\mapsto01\vert00011\vert100~\mbox{and}~\bfc'\mapsto 100\vert01\vert00011
    \end{equation*}
    produce an identical fragment multiset~$\{\{ 01,00011,100 \}\}$.
    \end{example}

\subsection{Existence}
The definition above immediately yields the following lemma.

\begin{lemma}\label{lemma:tBRC}
    If every pair of distinct codewords in $\cC\subseteq\{0,1\}^n$ is \emph{not} $t$-confusable, then $\cC$ is an $(n,t)$-BRC.
\end{lemma}
\begin{proof}
    Assume, for the sake of contradiction, that~$\cC$ is not~$(n,t)$-break-resilient.
    Then there exists a codeword $\bfc\in\cC$ that can be broken into $\le t+1$ fragments whose multiset does \emph{not} uniquely determine $\bfc$.  
    Hence there is another word $\bfc'\neq\bfc$ such that $\bfc'\in\cC$ and the same set of fragments can be reordered and concatenated to form both $\bfc$ and $\bfc'$.  
    But this means $\bfc$ and $\bfc'$ are $t$-confusable, contradicting the hypothesis.
\end{proof}

We continue by bounding the number of~$t$-confusable words of a given~$\bfc\in\bi^n$.
\begin{lemma}
    A word~$\bfc\in\bi^n$ is~$t$-confusable with at most~$\binom{n-1}{t}(t+1)!$ different words.
\end{lemma}
\begin{proof}
    For a word~$\bfc\in\bi^n$, there are~$\binom{n-1}{t}$ different ways to break it~$t$ times, i.e., inserting~$t$ bars in between of~$n$ stars.
    Each way yields a (not necessarily distinct) multiset of~$(t+1)$ fragments, and for one such multiset, there are~$(t+1)!$ different ways to permute its elements.
    Concatenating the permuted fragments gives a (not necessarily distinct) word.
    
    Note that we only consider breaking~$\bfc$~$t$ times, but no less, due to the following reason: if a word~$\bfc'$ can be generated using the above procedure with~$t'<t$ breaks, it can also be generated with~$t$ breaks by ``gluing''~$t-t'$ breaks before permuting the fragments.
    Hence, the above procedure generates all possible  words that are~$t$-confusable with~$\bfc$.
\end{proof}

We prove the existence of an~$(n,t)$-BRC using a Gilbert-Varshamov type argument, i.e., we begin with a candidate set for the code, and iteratively remove words from it until it becomes an~$(n,t)$-BRC.
Let a candidate set be the entire space of~$\bi^n$.
We construct an~$(n, t)$-BRC by repeating the following procedure until the candidate set is empty: choose an arbitrary word from the candidate set, and remove all its~$t$-confusable words.
Since at most~$ \binom{n-1}{t}(t+1)!$ words are removed during each iteration, we are guaranteed to obtain a set~$\cC$ of least~$2^n/[\binom{n-1}{t}(t+1)!]$ words before the procedure terminates.
The words are pair-wise not~$t$-confusable, and hence~$\cC$ is an~$(n, t)$-BRC by Lemma~\ref{lemma:tBRC}.

\begin{theorem}
    There exists an~$(n, t)$-BRC with redundancy~$O(t\log n)$.
\end{theorem}
\begin{proof}
    The above method generates an~$(n, t)$-BRC~$\cC$ of size at least~$2^n/[\binom{n-1}{t}(t+1)!]$, and hence its redundancy is at most
\begin{align*}
    n - \log\frac{2^n}{\binom{n-1}{t}(t+1)!} &= \log\left[\binom{n-1}{t}(t+1)!\right] =\log \frac{(n-1)!(t+1)!}{(n-t-1)!t!} \\
    &= \log(t+1) + \log (n-1)+\cdots +\log (n-t)=O(t\log n).\qedhere
\end{align*}
\end{proof}

\subsection{Converse}

The~$t$-confusability implies that certain constant-weight subcodes of~$\cC$ must have large Hamming distance.
\begin{lemma}\label{lemma:tconfusable}
    Let~$\cC\subseteq\bi^n$ be an~$(n, t)$-BRC, and let~$\cC=\cC_0\cup\cC_1,\cup\ldots\cup\cC_n$ be its partition to constant weight subcodes, i.e., where~$\cC_i$ contains all words in~$\cC$ of Hamming weight~$i$, for all~$i\in[n]$. Then, the minimum Hamming distance of each~$\cC_i$ is at least~$\ceil{\frac{t+1}{2}}$.
\end{lemma}

\begin{proof}
    Let~$i\in[n]$, and assume for contradiction that there exist~$\bfx,\bfy\in\cC_i$ such that~$d_H(\bfx,\bfy)\le\ceil{\frac{t+1}{2}}-1$.
    This implies that~$\bfx$ and~$\bfy$ are $t$-confusable, as demonstrated next.
    Consider the indices~$i_1,\ldots,i_{\ell}\in[n]$, for~$\ell\le\ceil{\frac{t+1}{2}}-1$, in which~$\bfx$ and~$\bfy$ differ.
    It follows that~$\bfx$ and~$\bfy$ can be written as
    \begin{align*}
        \bfx&=\bfc_1\circ x_{i_1}\circ\bfc_2\circ\ldots\circ \bfc_\ell\circ x_{i_\ell}\circ\bfc_{\ell+1}\\
        \bfy&=\bfc_1\circ y_{i_1}\circ\bfc_2\circ\ldots\circ \bfc_\ell\circ y_{i_\ell}\circ\bfc_{\ell+1}
    \end{align*}
    for some (potentially empty) strings~$\bfc_1,\ldots,\bfc_{\ell+1}$, where~$x_{i_j}\ne y_{i_j}$ for every~$j\in[\ell]$.
    Then, since~$2\ell\le t$ if~$t$ is even and~$2\ell\le t-1$ if~$t$ is odd, if follows that~$2\ell\le t$, and therefore an adversary may break either~$\bfx$ or~$\bfy$~$2\ell$ times.
    Specifically, consider an adversary which breaks~$\bfx$ and~$\bfy$ to the immediate left and the immediate right of entry~$i_j$, for each~$j\in[\ell]$.
    This produces the following multisets of fragments that are given to the decoder:
    \begin{align*}
        \cX &= \{\{\bfc_1,\ldots,\bfc_{\ell+1},x_{i_1},\ldots,x_{i_\ell}\}\}\mbox{, and}\\
        \cY &= \{\{\bfc_1,\ldots,\bfc_{\ell+1},y_{i_1},\ldots,y_{i_\ell}\}\}.
    \end{align*}
    In addition, since~$w_H(\bfx)=w_H(\bfy)=i$, where~$w_H$ denotes Hamming weight, it follows that~$w_H(x_{i_1},\ldots,x_{i_\ell})=w_H(y_{i_1},\ldots,y_{i_\ell})$, and hence the multisets~$\{\{x_{i_1},\ldots,x_{i_\ell}\}\}$ and~$\{\{y_{i_1},\ldots,y_{i_\ell}\}\}$ are identical.
    This implies that~$\cX=\cY$, and hence~$\bfx$ and~$\bfy$ are~$t$-confusable, a contradiction.
\end{proof}

By applying the classical sphere-packing bound~\cite{wikipedia_spherepacking} to those subcodes,
Lemma~\ref{lemma:tconfusable} gives rise to the following lower bound on the redundancy of an~$(n, t)$-BRC.

\begin{theorem}\label{theorem:bound}
    an~$(n, t)$-BRC~$\cC$ satisfies~$n-\log|\cC|\ge \Omega(t\log \tfrac{n}{t})$.
\end{theorem}
\begin{proof}
    For~$i\in[n]$ let~$\cC_i$ be the set of all codewords of~$\cC$ of Hamming weight~$i$, as in Lemma~\ref{lemma:tconfusable}, and let~$\cC_{i_\text{max}}$ be the largest set among the~$\cC_i$'s.
    We have that
    \begin{align*}
        \log|\cC|=\log(\textstyle\sum_{i=1}^n|\cC_i|)\le \log(n\cdot |\cC_{i_\text{max}}|)\le \log n+\log|\cC_{i_\text{max}}|.
    \end{align*}
    Furthermore, it follows from Lemma~\ref{lemma:tconfusable} that~$\cC_{i_\text{max}}$ is of minimum Hamming distance at least~$\ceil{\frac{t+1}{2}}$.
    Therefore, by the sphere packing bound we have that 
    \begin{align*}
        |\cC_{i_{\text{max}}}|\le \frac{2^n}{\sum_{j=0}^{t'}\binom{n}{j}}, \text{ where }t'\triangleq \bigg\lfloor\frac{\ceil{\frac{t+1}{2}}-1}{2}\bigg\rfloor\approx\frac{t}{4}.
    \end{align*}
    Hence, it follows that
    \begin{align*}
        n-\log|\cC|&\ge n-\log n-\log|\cC_{i_\text{max}}|\ge n-\log n-n+\log\left( \textstyle\sum_{j=0}^{t'}\binom{n}{j} \right)\\
        &\ge \log\binom{n}{t'}-\log n\overset{(\star)}{\ge}  \log((n/t')^{t'})-\log n=t'\log n-t'\log(t')-\log n=\Omega(t\log \tfrac{n}{t}),
    \end{align*}
    where~$(\star)$ follows since~$\binom{n}{t'}=\frac{n}{t'}\cdot\frac{n-1}{t'-1}\cdot\ldots\cdot \frac{n-t'+1}{1}$, and each of these terms is larger than~$n/t'$.
\end{proof}

Since we are mostly interested in the parameter regime where~$t$ is asymptotically smaller than~$n$, note that Theorem~\ref{theorem:bound} implies a minimum of~$\Omega(t\log n)$ redundant bits whenever~$t=O(n^{1-\epsilon})$ for any constant~$\epsilon>0$. 

\section{Code Construction}\label{section:code}

\begin{figure}[t]
	\centering
	\begin{subfigure}{0.49\textwidth}
		\includegraphics[width=1\textwidth]{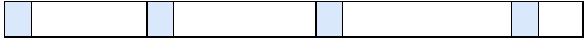}
		\caption{}
	\end{subfigure}
    \hfill
	\begin{subfigure}{0.49\textwidth}
		\includegraphics[width=1\textwidth]{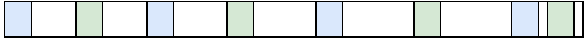}
		\caption{}
	\end{subfigure}\\
    \centering
	\begin{subfigure}{0.49\textwidth}
		\includegraphics[width=1\textwidth]{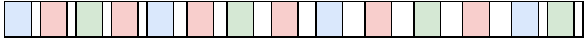}
		\caption{}
	\end{subfigure}
        \hfill
	\begin{subfigure}{0.49\textwidth}
		\includegraphics[width=1\textwidth]{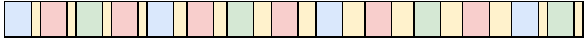}
		\caption{}
	\end{subfigure}
 	\caption{Illustration of the inductive identification of \emph{beacons} and \emph{residuals} on a uniformly random string~$\bfz$.
    (a) Level-$0$ beacons (cyan) are all occurrences of codewords from a predetermined mutually-uncorrelated code~$\cC_\text{MU}$ of length~$n_\text{MU}$.
    (b)-(c) Level-$1$ (green) and level-$2$ (red) beacons are defined inductively: for every pair of adjacent beacons already defined, we mark the~$n_\text{MU}$-bit substring which lies in the middle between them as a beacon of the new level.
    (d) Once no further layers can be added, the remaining substrings (yellow) between adjacent beacons are designated as \emph{residuals}.
    }
    \label{fig:beacons-and-residuals}
\end{figure}

\subsection{Overview}~\label{section:overview}
Many synchronization problems in coding theory (e.g.,~\cite{cheng2018deterministic,shomorony2021torn,bar2023adversarial}), where a sequence of received bits must be aligned against an original reference sequence, rely on matching short, easily recognizable \emph{beacon} strings embedded \emph{at fixed positions} within the codeword.
However, the use of \emph{beacons at fixed positions} comes at a heavy cost in terms of redundancy, since the beacon bits by themselves contain no information.
To remedy this, we employ the observation of~\cite{cheng2018deterministic}, that \emph{naturally occurring} patterns in a random string can be utilized as beacons.
That is, a uniformly random string is likely to contain numerous instances of carefully defined patterns that can serve as beacons, thereby enabling the alignment of unordered fragments with the original information.
This insight enables the construction of a break-resilient code with low redundancy from randomly sampled binary strings.
 
Intuitively, the beacons in our scenario should satisfy two conditions: (i) they are short and easily recognizable, and (ii) they do not overlap.
These two properties enable the decoder to unambiguously locate every beacon which survives the breaks, and guarantee that any single break can destroy at most one beacon (since if beacons overlap, a break within the overlapping segment would eliminate all of them).
In what follows we present a unique way of identifying naturally occurring beacons in any given~$\bfz\in\{0,1\}^m$. 
This method, which requires one to fix a mutually uncorrelated code~$\cC_{\text{MU}}$ a priori, will be the basis for the construction of our~$(n,t)$-BRC.
For ease of exposition we provide high-level details in this section, and full formal definitions are given in Section~\ref{section:encoding} and Section~\ref{section:decoding}.

\begin{definition}[level-$0$ beacon]\label{definition:l0beacons}
    Let~$\cC_\text{MU}$ be a mutually uncorrelated code (see Section~\ref{section:MU}) of length~$n_\text{MU}$.
    For~$\bfz\in\bi^m$, where~$m\gg n_\text{MU}$, the multiset of level-$0$ beacons is
    \begin{equation}\label{eq:l0beacons}
        \cS_0 \triangleq \{\mbox{all substrings of~$\bfz$ which are codewords of~$\cC_\text{MU}$}\}.
    \end{equation}
\end{definition}

With this choice, the decoder can detect every beacon in~\emph{any} fragment simply by sliding a window of length $n_\text{MU}$ over the fragment and checking membership in $\cC_\text{MU}$.
Moreover, since codewords in $\cC_{\text{MU}}$ are mutually uncorrelated, no two level-$0$ beacons can overlap.
As a result, each intact beacon can be detected in a fragment without ambiguity.
Since we depend on these beacons to infer the correct ordering of fragments, it is advantageous to embed as many of them as possible.
Accordingly, we now introduce higher-level beacons to enrich the beacon structure.
\begin{definition}[level-$\ell$ beacon]\label{definition:llbeacons}
    For~$\ell>0$, assume~$\cS_0,\ldots,\cS_{\ell-1}$ are already defined on~$\bfz$.
    A level-$\ell$ beacon is any substring of length~$n_\text{MU}$ whose starting position is the midpoint between the starting points of two adjacent beacons in~$\cS_0\cup\ldots\cup\cS_{\ell-1}$, and does not overlap with either of them.
\end{definition}
That is, a level-$\ell$ beacon lies in the middle of two adjacent beacons from level~$0$ through~$(\ell-1)$, if they are at least~$n_\text{MU}$ bits apart.
In cases where two adjacent beacons are too close for a new-level beacon to fit between them, we denote the bits in between as a \emph{residual}.
By induction, beacons at all levels remain pairwise disjoint, partitioning the entire string $\bfz$ into beacons and residuals.
Figure~\ref{fig:beacons-and-residuals} illustrates this inductive definition; we further note that the partitioning of any string~$\bfz\in\{0,1\}^m$ to beacons and residuals is \textit{unique}.

The success of decoding relies on the following key observation: if (a) every pair of beacons in $\bfz$ is distinct (i.e.,~$\cS_0$ in~\eqref{eq:l0beacons} forms a set rather than a multiset), and (b) both the identity of a beacon and its position in~$\bfz$ are known,
then the beacon can serve as an  ``anchor'', allowing the fragment containing it to be placed unambiguously to its correct position in~$\bfz$.
Building on this observation, the decoding process proceeds recursively.

The decoder begins with an incomplete reconstruction~$\bfz' \in\bi^m$, in which only level-$0$ beacons are placed at their precise positions as in~$\bfz$; how the decoder acquires this initial information will be detailed in the later sections.
It then examines each fragment and anchors every fragment containing at least one such beacon to~$\bfz'$.
These newly anchored fragments partially reveal level-$1$ beacons and their respective positions, enabling the next iteration of anchoring.
This recursive procedure continues until no further fragments can be placed, and then the decoding process concludes with the recovery of all residual symbols.
A schematic illustration of this process is provided in Figure~\ref{fig:affix}.

The central technical questions are: (a) how can the decoder reliably identify level-$0$ beacons and their exact positions initially, and
(b) how can it progressively discover all level-$\ell$ beacons using only the partial information revealed by fragments anchored in previous iterations, despite the adversarial nature of the channel?
Addressing these questions constitutes the core of our code construction, whose formal details are developed in the subsequent sections.

\subsection{Encoding}\label{section:encoding}

Let~$t$ be the security parameter,  and let~$m>t$ be the length of a uniformly random string~$\bfz$.
Let~$\cC_\text{MU}$ be a mutually uncorrelated code whose code length is~$n_\text{MU}=c\log m$, where~$c\geq 3$ is a constant (but not necessarily an integer).
The code length, size, and redundancy of our~$(n,t)$-BRC $\cC$ are functions of~$t$ and~$m$, and will be discussed in Section~\ref{section:redundancy}.
Since \cite{levy2018mutually} provides an efficient construction of binary MU code of any length~$n_{\text{MU}}$ with no more than~$\ceil{\log (n_\text{MU})}+4$ redundant bits (see Section~\ref{section:MU}), which is less than~$\log (n_\text{MU})+5$, we have the code size
\begin{equation}
    \vert\cC_\text{MU}\vert\geq\frac{2^{n_\text{MU}}}{2^{\log(n_\text{MU})+5}}=\frac{2^{c\log m}}{\beta c\log m},~\mbox{where~$\beta=2^5=32$}.
\end{equation}

Our codeword construction begins by selecting a uniformly random string~$\bfz\in\bi^m$ to serve as the information word, while rejecting and resampling if~$\bfz$ does not satisfy several criteria.
These criteria enable the encoder to leverage the structure of the beacons to generate carefully designed redundancy bits, transforming~$\mathbf{z}$ into a $(n,t)$-break-resilient codeword.

First, we require that adjacent level-$0$ beacons are not too far apart.
This constraint limits the number of beacon levels, which in turn bounds the redundancy of the code, as we will show in sequel.
Second, we require that every beacon appears \emph{exactly once} in~$\mathbf{z}$ to avoid ambiguity during decoding due to the reason mentioned in the previous section.
Finally, we designate $t+1$ lexicographically first codewords~$\mathbf{m}_0, \ldots, \mathbf{m}_t \in \mathcal{C}_{\text{MU}}$ as \emph{markers}.
These markers will later be used to distinguish redundancy bits from information bits, and therefore must not appear in $\mathbf{z}$.

Together, these conditions define what we call a \emph{legit} string, formalized below.

\begin{figure}[t]
	\centering
	\begin{subfigure}{0.49\textwidth}
		\includegraphics[width=1\textwidth]{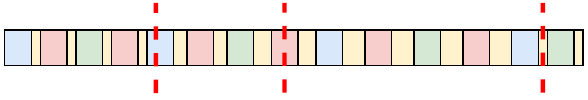}
		\caption{}
	\end{subfigure}
    \hfill
	\begin{subfigure}{0.49\textwidth}
		\includegraphics[width=1\textwidth]{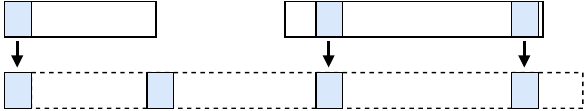}
		\caption{}
	\end{subfigure}\\
    \centering
	\begin{subfigure}{0.49\textwidth}
		\includegraphics[width=1\textwidth]{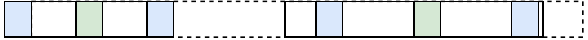}
		\caption{}
	\end{subfigure}
        \hfill
	\begin{subfigure}{0.49\textwidth}
		\includegraphics[width=1\textwidth]{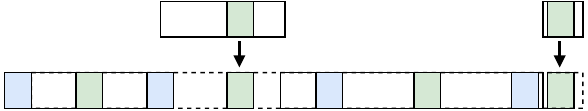}
		\caption{}
	\end{subfigure}
 	\caption{(a) The string~$\bfz$ (Figure~\ref{fig:beacons-and-residuals}) is adversarially broken along the dashed red lines, and the fragments are received unordered at the decoder.
    (b) At the first step, the decoder locates all fragments which contain at least one level-$0$ beacon. 
    These can be anchored to their original position according to redundant information that is appended to the codeword; this redundant information, which is much shorter than the codeword itself, is protected in more wasteful means.
    (c) The midpoints between anchored level-$0$ beacons are level-$1$ beacons, revealing partial information about the ordering of \emph{all} level-$1$ beacons.
    (d) The partial information about the ordering of level-$1$ beacons enables any fragment which contains an intact level-$1$ beacon to be anchored to its correct position. 
    This process repeats with level-$2$, level-$3$ beacons, etc., until no more fragments can be anchored. Note that the recovery of residuals are not included in the figures.
    }
    \label{fig:affix}
\end{figure}

\begin{definition}\label{def:legit}
A binary string~$\bfz\in\bi^m$ is called~\emph{legit} if it satisfies the following properties.

\begin{enumerate}[start=1,label={(\Roman*)}]
    \item \label{item:distance}
    Every interval of~$(2\beta c\log^2m+c\log m-1)$ bits in~$\bfz$ contains a level-$0$ beacon.
    \item \label{item:distinct}
    Every two non-overlapping substrings of length~$c\log m$ of~$\bfz$ are distinct. 
    \item \label{item:no-markers}
    $\bfz$ does not contain any of the markers~$\bfm_0,\ldots,\bfm_t$. 
\end{enumerate}    
\end{definition}

In a nutshell, an~$(n, t)$-BRC codeword~$\bfc\in\cC$ is constructed by attaching recursively generated redundancy bits to a legit string~$\bfz$.
Specifically, for level-$\ell$, let
\begin{equation*}
    \cup_{j=0}^\ell\cS_j= \{\bfs_1,\bfs_2,\ldots,\bfs_{r_\ell}\}
\end{equation*}
be the set of all~$r_\ell$ beacons from level~$0$ through~$\ell$, indexed by order of appearance in~$\bfz$, i.e.,~$\bfs_i$ is on the left of~$\bfs_j$ for pair~$i,j\in[r_l],i<j$.

\begin{algorithm}[t]\caption{Encoding}\label{alg:Encoding}
\begin{algorithmic}[1]
\Statex \textbf{Input:}~{A legit binary string $\bfz\in\bi^m$}
\Statex \textbf{Output:}~{A codeword $\bfc\in\bi^n$}
    \State Let~$\bfu_1,\ldots,\bfu_{t}$ be empty strings.
    \State $\bfy\gets \bfm_0\circ\bfz$\label{line:attachm0}
    \State Let~$\bfs_1,\ldots,\bfs_r$ be the $r$ level-$0$ beacons in~$\bfz$, let~$i_1,\ldots,i_r$ be their indices (in ascending order), i.e.,~$\bfs_j=\bfy[i_j,i_j+c\log m-1]$.
    \Statex \textcolor{red}{$\triangleright~\textsc{redundancy for level-$0$ beacons}$}
    \Commentx{$\texttt{BEACONS}$ maps an index~$i$ to~$\bfy[i,i+c\log m-1]$ if~$\bfy[i,i+c\log m-1]$ is a level-$0$ beacon.}
    \State Let~$\texttt{BEACONS}$ be a key-value store of size~$r$ such that~$\texttt{BEACONS}[i_j]=\bfs_j$ for all~$j\in[r]$. \label{line:l0beacons}
    \State Let~$\bfA=[A_{a,b}]\in\bbN^{\vert\cC_\text{MU}\vert\times \vert\cC_\text{MU}\vert}$ be an all-$0$ matrix.
    \ForAll{keys~$k$ in $\texttt{BEACONS}$ in ascending order}\label{line:createAdjacencyMatrix1}
        \State Let $k_\text{next}$ be the smallest key in $\texttt{BEACONS}$ larger than~$k$, or~$m+c\log m+1$ if~$k$ is the largest.
        \State Let~$a,b$ be the indices of~$\texttt{BEACONS}[k]$ and~$\texttt{BEACONS}[k_\text{next}]$ in~$\cC_\text{MU}$, respectively, when~$\cC_\text{MU}$ is ordered lexicographically.
        \State $A_{a,b}\gets k_\text{next}-k-c\log m$\Comment{$A_{a,b}\neq 0$ implies that two codewords of~$\cC_\text{MU}$ with indices~$a$ and~$b$}\label{line:createAdjacencyMatrix2}
        \Commentx{\hfill appear as adjacent level-$0$ beacons in~$\bfy$, separated by~$A_{a,b}$ many bits.}
    \EndFor
    \State $\texttt{COMP-A}\gets \textsc{compress-adjacency-matrix}(\bfA)$
    \Comment{Every row of~$\bfA$ is compressed to~$2c\log m$ bits, making a vector~\texttt{COMP-A} of~$\vert\cC_\text{MU}\vert$ elements.}
    \label{line:compressAdjacencyMatrix}
    \State $\bfd_1,\ldots,\bfd_{4t}\in\bbF_{2^{2c\log m}}\gets \textsc{rs-encode}(\texttt{COMP-A},4t)$\Comment{Every element of~\texttt{COMP-A} is treated as an element in~$\bbF_{2^{2c\log m}}$.}\label{line:encodeAdjacencyMatrix}
    \ForAll{$l\in[t]$}~$\bfu_l\gets\bfd_{4l-3}\circ\bfd_{4l-2}\circ\bfd_{4l-1}\circ\bfd_{4l}$\label{line:appendAdjacencyMatrixRedundancy}
    \EndFor
    \Statex \textcolor{red}{$\triangleright~\textsc{redundancy for higher-level beacons}$}
    \State Let~\texttt{NEW-BEACONS},~\texttt{RESIDUALS} be empty key-value stores.
    \For{$\texttt{level}\gets 1,\ldots,\log\log m+6$}
        \ForAll{keys~$k$ in $\texttt{BEACONS}$ in ascending order}
            \State Let $k_\text{next}$ be the smallest key in $\texttt{BEACONS}$ larger than~$k$, or~$m+c\log m+1$ if~$k$ is the largest.
            \If{$k_\text{next}-k\ge 2c\log m$}\label{noFit}\label{line:levelUp1}
                \State $u \gets (k+k_\text{next})/2$
                \State $\texttt{NEW-BEACONS}[u]\gets\bfy[u,u+c\log m-1]$\label{line:levelUp2}
            \ElsIf{$k_\text{next}-k>c\log m$}\label{line:noFit}\Comment{$k_\text{next}-k\le c\log m$ is impossible due the mutual uncorrelation of~$\cC_\text{MU}$.}
                \State $v\gets k+c\log m$
                \State $\texttt{RESIDUALS}[v]\gets\textsc{pad}(\bfy[v,k_\text{next}-1])$\label{line:storeResidualBits}                
            \EndIf
        \EndFor
        \State $\texttt{BEACONS} \gets \texttt{BEACONS} \cup \texttt{NEW-BEACONS}$        
        \State $\texttt{NEW-BEACONS}\gets \texttt{EMPTY}$
    \State Let~$k_1,k_2,\ldots$ be all keys in~\texttt{BEACONS} in increasing order.
        \State $\bfr_1,\ldots,\bfr_{2t}\gets \textsc{rs-encode}((\texttt{BEACONS}[k_1],\texttt{BEACONS}[k_2],\ldots),2t)$\label{line:encodebeacons}
      \ForAll{$l\in[t]$}~$\bfu_l\gets\bfu_l\circ\bfr_{2l-1}\circ\bfr_{2l}$\label{line:appendbeaconRedundancy}\EndFor
    \EndFor
    \Statex \textcolor{red}{$\triangleright~\textsc{redundancy for residuals}$}
    \State Let~$k_1,k_2,\ldots$ be all keys in~\texttt{RESIDUALS} in increasing order.
    \State $\bft_1\ldots,\bft_{3t}\gets\textsc{rs-encode}((\texttt{RESIDUALS}[k_1],\texttt{RESIDUALS}[k_2],\ldots),3t)$\label{line:encodeResiduals}
    \ForAll{$l\in[t]$}~$\bfu_l\gets\bfu_l\circ\bft_{3l-2}\circ\bft_{3l-1}\circ\bft_{3l}$\label{line:appendResidualRedundancy}\EndFor
    \Statex \textcolor{red}{$\triangleright~\textsc{instrumentation and final assembly}$}
    \State $\bfc\gets \bfy$
    \ForAll{$l\in[t]$}\label{line:finalAssemblyStart}
        \State $\bfc= \bfm_l\circ\bfu_l[1:c\log m/2]\circ\bfm_l\circ\bfu[c\log m/2+1:c\log m]\circ\bfm_l\cdots\circ\bfc$\Comment{instrumentation and assembly.}
    \EndFor\label{line:finalAssemblyEnd}
    \State \textbf{Output}~$\bfc$
\end{algorithmic}
\end{algorithm}
Denoting~$r_0$ by~$r$ for clarity, Property~\ref{item:distance} readily implies the following lemma, which is required in the sequel and is easy to prove.
    \begin{lemma}\label{lemma:speratedIntervals}
    A legit~$\bfz\in\{0,1\}^m$ can be written as
    \begin{equation*}
    \bfz=\bfz_1\circ\bfs_1\circ\bfz_2\circ\bfs_2\circ\cdots\circ\bfs_{r}\circ\bfz_{r+1}, 
    \end{equation*}
    where~$\bfz_1,\bfz_2,\ldots,\bfz_{r+1}$ are (potentially empty) intervals of~$\bfz$ separated by the level-$0$ beacons~$\bfs_1,\bfs_2,\ldots,\bfs_{r}$ (each of length~$c\log m$), and~$\vert\bfz_j\vert<2\beta c\log^2 m$ for every~$j\in[r+1]$.
\end{lemma}

A codeword~$\bfc$ is constructed by feeding a legit string~$\bfz$ to Algorithm~\ref{alg:Encoding} (sub-routines are described in Algorithm~\ref{alg:utilities} in Appendix~\ref{appendix:auxiliary}).
First, the algorithm attaches~$\bfm_0$ to the left hand side of~$\bfz$ (line~\ref{line:attachm0}); the resulting string~$\bfy$ is then called the~\emph{information region} of the codeword.
Note that~$\bfm_0$ serves two purposes here.
First, it becomes the first level-$0$ beacon in~$\bfy$, i.e.,~$\bfy[1:{c\log m}]=\bfm_0$. 
Second, since the algorithm will attach redundancy bits (i.e.,~\emph{redundancy region}) to the left of~$\bfy$, the string~$\bfm_0$ marks the transition between the information region and the redundancy region. 

\subsubsection{\textbf{Redundant bits for level-0 beacons}}
Recall that the recursive decoding procedure outlined in Section~\ref{section:overview} requires, as a prerequisite, that the decoder first identify all level-0 beacons along with their exact positions.
To ensure this, the encoder appends carefully designed redundancy bits that safeguard this initial information.
Specifically, it encodes both the pairwise ordering information and the distance between every two adjacent level-$0$ beacons.

The encoder creates a key-value (KV) store~\texttt{BEACONS}, which maps an index~$i$ of~$\bfy$ to the substring~$\bfy[i,i+c\log m-1]$ if the latter is a level-$0$ beacon, i.e., matches a codeword of~$\cC_\text{MU}$ (line~\ref{line:l0beacons}).
It then creates an all-$0$~$\vert\cC_\text{MU}\vert\times \vert\cC_\text{MU}\vert$ matrix~$\bfA$ to records the correct~\emph{pairwise ordering} and~\emph{pairwise positional distance}, of level-$0$ beacons.
Specifically, for each pair of adjacent beacons~$\bfs_j$ and~$\bfs_{j+1}$ in~$\bfy$, let~$a,b$ be the lexicographical orders of them as MU codewords in~$\cC_\text{MU}$, respectively.
The encoder update the element~$A_{a,b}$ of~$\bfA$ as the number of bits lies in between of~$\bfs_j$ and~$\bfs_{j+1}$ (line~\ref{line:createAdjacencyMatrix1}--\ref{line:createAdjacencyMatrix2}).

By Property~\ref{item:distinct} (a codeword of~$\cC_\text{MU}$ appears at most once in~$\bfz$) and Property~\ref{item:no-markers} (markers do not appear in~$\bfz$), every row of~$\bfA$ is either (1) the all-$0$ vector, or (2) a vector with exactly one non-zero entry, located in one of the~$\vert\cC_\text{MU}\vert-t$ positions (since~$\bfy$ is created by attaching~$\bfm_0$ to~$\bfz$, and~$\bfm_0,\ldots,\bfm_t$ do not to appear in the legit string~$\bfz$).

Consequently, there are fewer than~$\vert\cC_\text{MU}\vert$ admissible positions for the non-zero entry in each row, and by Lemma~\ref{lemma:speratedIntervals}, the number of possible values such a non-zero entry (if exists) is bounded above by~$2\beta c\log^2m$.
As such, there exist at most
$\vert\cC_\text{MU}\vert\cdot 2\beta c\log^2m\leq m^c\cdot m^c= m^{2c}$ different possible values for one individual row of~$\bfA$.
This fact enables to compress~$\bfA$ into~\texttt{COMP-A}, a vector in~$\bbF_{2^{2c\log m}}^{\vert\cC_\text{MU}\vert}$ (line~\ref{line:compressAdjacencyMatrix}).
This is done in the function~\textsc{compress-adjacency-matrix} (line~\ref{line:compress-adjacency-matrix}, Alg.~\ref{alg:utilities}), which applies function~\textsc{compress-row} to every individual row of~$\bfA$ (line~\ref{line:compressIndividualRow}, Alg.~\ref{alg:utilities}).
Note that the function~\textsc{compress-row} is simple to implement using indexing methods and its details are omitted for brevity.
The vector~\texttt{COMP-A} is later encoded using a systematic Reed-Solomon code to produce~$4t$ parity symbols~$\bfd_1,\ldots,\bfd_{4t}\in\bbF_{2^{2c\log m}}$\footnote{
The encoding is performed over the field of size~$m^{2c}$, and requires at least~$\vert\cC_\text{MU}\vert+4t$ distinct field elements; this is the case since~$m^{2c}- \vert\cC_\text{MU}\vert-4t\geq m^{2c}-m^c-t\geq 0$.} (line \ref{line:encodeAdjacencyMatrix}), and
every group of~$4$ is gathered and attached to a redundancy string in~$\bfu_1,\ldots,\bfu_t$ (line~\ref{line:appendAdjacencyMatrixRedundancy}).

\subsubsection{\textbf{Redundancy for higher-level beacons}}
Since the initial information specifies the exact positions of all level-$0$ beacons, the positions of higher-level beacons are inherently determined by their structure detailed in Definition~\ref{definition:llbeacons}.
Consequently, the encoder needs only to protect the identities of the beacons at each level~$\ell>0$, rather than their positions.

During recursive step~$\ell\ge 1$, the algorithm adds level-$\ell$ beacons to \texttt{BEACONS}, by locating the midpoint between every two existing beacons (line~\ref{line:levelUp1}--\ref{line:levelUp2}).
If the interval between two consecutive beacons is too short to contain a beacon (line~\ref{line:noFit}), it is padded to a string of length~$c\log m$, stored in a KV store~\texttt{RESIDUALS} (line~\ref{line:storeResidualBits}), and referred to as a~\emph{residual}.
The padding is performed by attaching a~$1$ and sufficiently many~$0$'s until the padded interval is~$c\log m$ bits long (line~\ref{line:paddingFunction}, Alg.~\ref{alg:utilities}).
For the sake of encoding and decoding, it is important to note that this padding operation is injective and reversible.

The beacons in level-$1$ through level-$\ell$, stored in~\texttt{BEACONS}, are viewed as a vector sorted in ascending order of the keys, and then encoded with a Reed-Solomon code to generate~$2t$ redundancy symbols\footnote{
The encoding is performed over the field of size~$2^{c\log m}=m^{c}$, each in~$\bbF_{2^{c\log m}}$.
Since there are at most~$m/(c\log m)$ residuals/beacons in~$\bfz$, the encoding requires~$m/{(c\log m})+2t$ distinct field elements.
The encoding is feasible due to the fact that~$m^{c}- m/(c\log m)-2t\geq m^{c-1}-2t\geq m^{c-2}-t\geq 0$.} (line~\ref{line:encodebeacons}).
The iteration proceeds until no new beacons can fit between any two adjacent existing beacons.
Due to Lemma~\ref{lemma:speratedIntervals}, at most~$2\beta \log m-1$ non-overlapping beacons, each of length~$c\log m$, can fit in between any two level-$0$ beacons.
Meanwhile, the total number of beacons inside such interval is~$2^{\ell}-1$ after~$\ell$ iterations, since every level-$\ell$ beacon is located in the middle of intervals separated by beacons from level~$0$ through~$(\ell-1)$.
As a result, the iteration is guaranteed to terminate after~$\log (2\beta\log m)=\log\log m+6$ levels.

\begin{figure}[t]
    \centering
    \input{figures/instrumentation}
    \caption{The redundancy string~$\bfu_l$ is instrumented with marker~$\bfm_l$. 
    First,~$\bfu_l$ is segmented into intervals of length~$c\log m/2$.
    Then, marker~$\bfm_l$ is inserted before each interval.}
    \label{fig:Instrumentation}
\end{figure}

\subsubsection{\textbf{Redundancy for residuals}}
Finally, the residuals are encoded similar to the beacons,
producing~$3t$ redundant symbols~$\bft_1,\ldots,\bft_{3t}$ (line~\ref{line:encodeResiduals}), and appended to the redundancy strings~$\bfu_1,\ldots,\bfu_{t}$ (line~\ref{line:appendResidualRedundancy}). 
Now, each redundancy string~$\bfu_l$ is of length~$(6+2\log\log m)\cdot c\log m$.
It is then instrumented in between every interval of~${c\log m}/{2}$ bits using the respective marker~$\bfm_l$, making it~$(6+2\log\log m)\cdot 3c\log m$ bits long (see Fig.~\ref{fig:Instrumentation}).
The redundancy strings are attached one by one to the left of the information region~$\bfy$ to create the final codeword~$\bfc$ (line~\ref{line:finalAssemblyStart}--\ref{line:finalAssemblyEnd}) of length
\begin{equation*}
    n\triangleq\vert\bfc\vert =m+ (6+2\log\log m)\cdot 3c\log m\cdot t +c\log m.
\end{equation*}

\begin{algorithm}\caption{Decoding, Part 1}\label{alg:Decode1}
\begin{algorithmic}[1]
\Statex \textbf{Input:}~{A multiset $\texttt{FRAGMENTS}$ of at most~$t+1$ (unordered) fragments of some codeword~$\bfc\in\cC$.} 
\Statex \textbf{Output:}~{The binary string~$\bfz$ such that Algorithm~\ref{alg:Encoding} with input~$\bfz$ yields $\bfc$.}
\Statex \textcolor{red}{$\triangleright~\textsc{classification of fragments}$}
\If{there exists $\bff\in\texttt{FRAGMENTS}$ and index~$i$ such that~$\bfm_0=\bff[i:i+c\log m-1]$}
    \State Remove~$\bff$, then add $\bff[0:i-1]$ and $\bff[i:]$ to \texttt{FRAGMENTS}.\label{line:transpoint}
\EndIf
\State Let~\texttt{R-FRAGMENTS} be the set of fragments that contains~$\bfm_l$ for some~$l\in[0,t]$ (fragments in the redundancy region).\label{line:R-FRAGMENTS}
\State Let~\texttt{Z-FRAGMENTS} be the remaining fragments that are either at least~$3c\log m$ bits, or contain at least one discernible level-$0$ beacon, i.e., a codeword of~$\cC_\text{MU}$ (fragments in the information region). 
\label{line:Z-FRAGMENTS}
\Statex \textcolor{red}{$\triangleright~\textsc{extract level-$0$ beacons and redundancy strings}$}
\State Let~$\bfA'=[A'_{a,b}]\in\bbN^{\vert\cC_\text{MU}\vert\times \vert\cC_\text{MU}\vert}$ be an all-$0$ matrix.
\ForAll{$\bff$ in \texttt{Z-FRAGMENTS}}\label{line:ZFRAGMENTBEGIN}
    \ForAll{level-$0$ beacons with index~$i\in[|\bff|]$ in~$\bff$ in ascending order}
        \If{$i_\text{next}$ exists, defined as the smallest index of a level-$0$ beacon that is greater than~$i$}
            \State Let~$\bfc_a\triangleq \bff[i,i+c\log m-1]$ and $\bfc_b\triangleq\bff[i_\text{next},i_\text{next}+c\log m-1]$, where and~$\bfc_a,\bfc_b\in\cC_\text{MU}$.
            \State $A'_{a,b}\gets i_\text{next}-i-c\log m$.
        \EndIf
    \EndFor\label{line:ZFRAGMENTEND}
\EndFor

\State Let~$\texttt{R-STRINGS}$ be an empty KV store.
\ForAll{$\bff$ in \texttt{R-FRAGMENTS}}\label{line:RFRAGMENTBEGIN} 
    \ForAll{$\bfu'_l$, defined as consecutive length-$c\log m/2$ intervals separated by~$2\cdot(6+2\log\log m)$ occurrence of~$\bfm_l$s}
        \State $\texttt{R-STRINGS}[l]\gets\bfu'_l.\textsc{remove}(\bfm_l)$\label{line:RFRAGMENTEND} \Comment{Remove the~$\bfm_l$'s instrumented in line~\ref{line:finalAssemblyEnd}, Alg.~\ref{alg:Encoding}}
    \EndFor
\EndFor

\Statex \textcolor{red}{$\triangleright~\textsc{recovery of level-$0$ beacons}$}
\State $\texttt{APPROX-COMP-A}\gets \textsc{compress-adjacency-matrix}(\bfA')$
\State $\texttt{COMP-A}\gets\textsc{repair-adj-matrix}(\texttt{APPROX-COMP-A},\texttt{R-STRINGS})$. \label{line:repairHisto}
\State $\bfA\gets \textsc{decompress-adjacency-matrix}(\texttt{COMP-A})$
\State \label{line:allocateyp}Let~$\bfy'\gets \bfm_0\circ*^{|\bfz_1|}\circ\bfs_1\circ*^{|\bfz_2|}\circ\bfs_2\circ\cdots\circ\bfs_{r}\circ*^{|\bfz_{r+1}|}$\Comment{$\bfz_1,\bfz_2,\ldots,\bfz_{r+1}$ are defined in Lemma~\ref{lemma:speratedIntervals}}
\end{algorithmic}
\end{algorithm}
\subsection{Decoding}\label{section:decoding}
A procedure for recovering the correct legit string~$\bfz$ from at most~$t+1$ fragments of the respective codeword~$\bfc$ is presented in Algorithm~\ref{alg:Decode1} and Algorithm~\ref{alg:Decode2}.
We now detail the procedure and establish its correctness.

\subsubsection{\textbf{Classification of fragments}}
Recall that a codeword~$\bfc$ consists of an information region and a redundancy region; the former is composed of~$\bfm_0$ followed by the legit string~$\bfz$, and the latter includes~$t$ instrumented redundancy string~$\bfu_1,\ldots,\bfu_t$ (see Figure~\ref{fig:Instrumentation}).
The decoder begins with classifying the fragments into two regions, and it first attempts to distinguish the transition point, i.e.~$\bfm_0$, from all fragments.

If a fragment containing~$\bfm_0$ is found, it is broken at~$\bfm_0$ such that the two segments belong to different regions (line~\ref{line:transpoint}).
Then, the decoder classifies the fragments into those belonging to the redundancy region (line~\ref{line:R-FRAGMENTS}), and those belonging to the information region (line~\ref{line:Z-FRAGMENTS}).
The purpose of this partition is to separate the treatments of these fragments---the information fragments need to be analyzed in order to extract an approximate adjacency matrix, whereas the redundancy fragments need to be analyzed in order to extract the Reed-Solomon redundancy symbols required for error correction. 
Such classification guarantees the following.
\begin{lemma}\label{lemma:insideInfoRegion}
    Every fragment~$\bff\in\texttt{Z-FRAGMENTS}$ is entirely contained in the information region of~$\bfc$.
\end{lemma}
\begin{proof}
    Assume otherwise, i.e., that there exists a fragment~$\bff\in\texttt{Z-FRAGMENTS}$ that intersects nontrivially with the redundancy region of the codeword~$\bfc$.
    Such~$\bff$ must not contain any of~$\bfm_1,\ldots,\bfm_t$, or it would have been stored in~\texttt{R-FRAGMENTS} (line~\ref{line:R-FRAGMENTS}, Alg.~\ref{alg:Decode1}). 
    Hence, it may reside in~\texttt{Z-FRAGMENTS} due to exactly one of the following reasons:
    \begin{enumerate}
        \item $\bff$ contains a discernible level-$0$ beacon.
        \item $\bff$ does not contain a discernible level-$0$ beacon, but $\vert\bff\vert \geq 3c\log m$.
    \end{enumerate}
    We proceed to show that such~$\bff$ does not exist for either reason.    
    Consider the first reason, i.e., $\bff$ contains a discernible level-$0$ beacon.
    In this case, since level-$0$ beacons are entirely contained in the information region, it follows that~$\bff$ must contain~$\bfm_0$.
    However, line~\ref{line:transpoint} assures that for such a fragment,~$\bfm_0=\bff[1:c\log m]$, and as a result,~$\bff$ does not intersect with the redundancy region, a contradiction.

    We proceed to the second reason.
    In this case, such a fragment~$\bff$ must intersect with both regions, as we show by contradiction as follows.
    Recall that the redundancy region is created by instrumenting the redundancy strings for every interval of length~$c\log m/2$ (see Fig.~\ref{fig:Instrumentation}).
    If~$\bff$ is entirely contained in the redundancy region, then~$\bff$ must contain~$\bfm_l$ for some~$l\in[t]$ due to the fact that~$\vert\bff\vert\geq 3c\log m$.
    A contradiction.

    However, such an~$\bff$ that intersect with both regions cannot exist.
    Let the information region of~$\bfc$ be~$\bfc[i_\text{trans}:]$, where~$\bfc[i_\text{trans}:i_\text{trans}+c\log m-1]=\bfm_0$ and~$\bfc[i_\text{trans}+c\log m:]=\bfz$.
    Then, since~$\bff$ intersects both regions, it follows that~$\bff=\bfc[i:j]$ for some~$i,j$ where~$i<i_\text{trans}$ and~$j\geq i_\text{trans}$.
    If~$j<i_\text{trans}+c\log m-1$, then since~$|\bff|\ge3c\log m$, it follows that~$i$ must be less than~$i_\text{trans}-2c\log m$, and hence must contain a copy of the marker~$\bfm_t$, a contradiction by line~\ref{line:R-FRAGMENTS}.
    If~$j\geq i_\text{trans}+c\log m-1$, then~$\bff$ contains~$\bfm_0$, and a contradiction ensues due to line~\ref{line:transpoint}.
\end{proof}

\begin{lemma}\label{lemma:atLeastOneSig}
    Every fragment~$\bff\in\texttt{Z-FRAGMENTS}$ contains at least one beacon in some level.
\end{lemma}
\begin{proof}
    By Lemma~\ref{lemma:insideInfoRegion}, every fragment~$\bff\in\texttt{Z-FRAGMENTS}$ is entirely contained in the information region of~$\bfc$.
    Recall that~$\bff$ either contains a discernible level-$0$ beacon, or is at least~$3c\log m$ bits long.
    Observe that beacons (of length~$c\log m$) are separated by residuals (of length at most~$c\log m-1$), and as a result,~$\bff$ must contain a beacon in the latter case (i.e.,~$\vert\bff\vert\geq 3c\log m$).
\end{proof}

\begin{lemma}\label{lemma:canBeFound}
    Every level-$0$ beacon, if not broken, can be found from fragments in the set~\texttt{Z-FRAGMENTS}.
\end{lemma}
\begin{proof}
    Notice that the decoding  algorithm first puts every fragment that contain a marker in~\texttt{R-FRAGMENTS} (line~\ref{line:R-FRAGMENTS}).
    The remaining fragments, as long as they contain a discernible level-$0$ beacon, will be stored in~\texttt{Z-FRAGMENTS} (line~\ref{line:Z-FRAGMENTS}).
    Therefore, if an unbroken level-$0$ beacon is missing from~\texttt{Z-FRAGMENTS}, it is contained in a fragment that contains a marker and was therefore placed in~\texttt{R-FRAGMENTS}.
    Such a fragment must cross~$\bfz$ and the redundancy region, and as result, must contain~$\bfm_0$. 
    However, line~\ref{line:transpoint} assures that it will be segmented and the right part, which contains the beacon, will be placed in~\texttt{Z-FRAGMENTS}.  
\end{proof}

The three preceding lemmas form the basis of the decoding process that described next.

\subsubsection{\textbf{Recovery of level-0 beacons}}
The decoding proceeds to the recovery of level-$0$ beacons, including their identities and positions.
A preliminary analysis of the~$\bfz$-fragments, whose purpose is to extract the surviving level-$0$ beacons
into an approximate adjacency matrix~$\bfA'$, is given in lines~\ref{line:ZFRAGMENTBEGIN}-\ref{line:ZFRAGMENTEND}.
In this analysis, all codewords of~$\cC_\text{MU}$ present in the fragments are located and identified as level-$0$ beacons, and the collection of all level-$0$ beacons is coalesced into a pair-wise ordering in the form of an \emph{approximate} adjacency matrix~$\bfA'$. 
Note that the number of errors in~$\bfA'$ relative to~$\bfA$ is bounded by the following lemma.

\begin{lemma}\label{lemma:boundErrors}
    Let~$t_1$ be the number of breaks that fall in the information region of the codeword, and let~$\bfA$ be the adjacency matrix of all level-$0$ beacons in~$\bfz$.
    Then,~$\bfA'$ differs from~$\bfA$ at no more than~$2t_1$ positions.
\end{lemma}
\begin{proof}
    By Lemma~\ref{lemma:canBeFound}, the decoding algorithm may fail to identify a level-$0$ beacon in the information region~$\bfy$ if there exists a break inside it.
    Hence, the algorithm fails to identify at most~$t_1$ level-$0$ beacons.
    Notice that, failing to capture a level-$0$ beacon~$\bfs=\bfc_a$ (i.e., a codeword in~$\cC_\text{MU}$ with lexicographical order~$a$) affects exactly two rows of~$\bfA$.
    That is, row~$a$ and another row whose~$a$-th element is non-zero.
\end{proof}

The analysis of the redundancy fragments is conducted in lines \ref{line:RFRAGMENTBEGIN}-\ref{line:RFRAGMENTEND}, during which all markers~$\bfm_l$, if not broken, are identified.
Recall that each marker~$\bfm_l$ is inserted to the redundancy string~$\bfu_l$ between every interval of~$c\log m /2$ bits, and~$\vert\bfu_l\vert=(6+2\log\log m)\cdot  c\log m$.
Hence, the redundancy string~$\bfu_l$ can be identified by observing a series of~$2\cdot (6+2\log\log m)$ markers~$\bfm_l$, each separated by~$c\log m/2$ bits.
All extracted redundancy strings are placed in a KV-store \texttt{R-STRINGS}.

\begin{lemma}\label{lemma:enoughRedundancy}
    Let~$t_2$ be the number of breaks that fall in the redundancy region of the codeword.
    Then, the decoding algorithm is guaranteed to obtain~$t-t_2$ redundancy strings.
\end{lemma}
\begin{proof}
    Recall that for every~$l\in[t]$, the redundancy string~$\bfu_l$ is instrumented in between of every two intervals of~$c\log m/2$ bits using the marker~$\bfm_l$.
    As a result of this instrumentation, every substring~$\bfs$ of length~$c\log m$ of a fragment in~\texttt{R-FRAGMENTS} must either be one of the~$\bfm_l$'s inserted for the sake of instrumentation, or not belong to~$\cC_\text{MU}$ altogether.
    Otherwise,~$\bfs$ overlaps a prefix or suffix of another codeword of~$\cC_\text{MU}$ (possibly itself), a contradiction.
    
    Therefore, every marker identified by the decoder is instrumented by the encoder, and does not contain any bits of redundancy strings (before instrumentation).
    As such, the decoder can obtain a redundancy string~$\bfu_l$ after identifying~$2\cdot(6+2\log\log m)$ consecutive occurrences of~$\bfm_l$.
    Since there are~$t$ (instrumented) redundancy strings, at least~$t-t_2$ out of them are intact, i.e., not broken, and are guaranteed to be found and stored in~\texttt{R-STRINGS}.
\end{proof}

The decoding algorithm proceeds to correct the constructed adjacency matrix~$\bfA'$ to~$\bfA$, i.e., the~\emph{correct} redundancy matrix generated in Algorithm~\ref{alg:Encoding} from~$\bfz$, using the collected redundancy strings (line~\ref{line:repairHisto}) and a standard Reed-Solomon decoder.
The success of the decoding process is guaranteed as follows.
\begin{theorem}\label{theorem:recoverAdjacencyMatrix}
    Line~\ref{line:repairHisto} outputs the correct adjacency matrix~$\bfA$.
\end{theorem}
\begin{proof}
 
    In~\textsc{repair-adj-matrix} (line~\ref{line:func-repair-adj-matrix}, Alg.~\ref{alg:utilities}), a codeword is constructed by coalescing the elements in~\texttt{APPROX-COMP-A} with redundancy symbols in~\texttt{R-STRINGS}, which is then fed into a Reed-Solomon decoder (line~\ref{line:decodeAdjMatrix}, Alg.~\ref{alg:utilities}).
    By Lemma~\ref{lemma:enoughRedundancy}, \texttt{R-STRINGS} contains at least~$t-t_2$ non-empty entries, and as a result, the constructed codeword contains less than $4t_2$ erasures (from the~\texttt{empty} entries in~\texttt{R-STRINGS}).
    Meanwhile, by Lemma~\ref{lemma:boundErrors} the number of rows in which~$\bfA$ and~$\bfA'$  differ is bounded by~$2t_1$.
    Since the compression of~$\bfA$ and $\bfA'$ (line~\ref{line:compressIndividualRow}, Alg.~\ref{alg:utilities}) collapses every row into one extension-field element, it follows that the compressed versions of~\texttt{COMP-A} and~\texttt{APPROX-COMP-A} also differ by at most~$2t_1$ entries. 
    Hence, the systematic part of the constructed codeword has at most~$2t_1$ errors.

    Recall that the encoding process generated~$4t$ redundant symbols from~\texttt{COMP-A}, and hence the decoding in line~\ref{line:decodeAdjMatrix} of Alg.~\ref{alg:utilities} is guaranteed to be successful since
    \begin{equation*}
       2t_1\cdot2+4t_2\leq 4(t_1+t_2)\leq 4t,
    \end{equation*}
    where the last transition follows the actual number of breaks~$t_1+t_2$ is at most the security parameter~$t$.
    The proof is concluded since a~$(k+4t,k)$ Reed-Solomon code can simultaneously correct any~$a$ errors and any~$b$ erasures as long as~$4t\ge 2a+b$.
\end{proof}

\begin{algorithm}[t]\caption{Decoding, Part 2}\label{alg:Decode2}
\begin{algorithmic}[1]
\Statex \textcolor{red}{$\triangleright~\textsc{recovery of higher level-beacons}$}
\State Let~\texttt{BEACONS} be a KV store s.t. for all level-$0$ beacon~$\bfs_j=\bfy'[i_j:i-j+c\log m-1]$, it holds that~$\texttt{BEACONS}[i]=\bfs$.\label{line:decode-kv-beacons}
\State Let~\texttt{UPDATED-BEACONS} be an empty KV store, and let~$\texttt{level}\gets0$\label{line:decode-new-beacons}
\State $\texttt{Z-FRAGMENTS},\bfy'\gets\textsc{anchor-fragments}(\texttt{BEACONS},\texttt{Z-FRAGMENTS},\bfz)$\label{line:affixz-fragments}

\While{\texttt{Z-FRAGMENTS} is not empty}\label{line:WhileBegin}
    \State $\texttt{level}\gets \texttt{level}+1$
    \ForAll{keys~$i$ in $\texttt{BEACONS}$ in ascending order}\label{line:ForBegin}
        \State Let~$i_\text{next}$ be the smallest key greater than~$i$, or~$m+c\log m+1$ if~$i$ is the greatest.
        \State $\texttt{UPDATED-BEACONS}[i]\gets\texttt{BEACONS}[i]$
        \If{$i_\text{next}-i\ge 2c\log m$}
            \State $u \gets (i+i_\text{next})/2$\label{line:middlePoint}
                \If{$\bfy'[u,u+c\log m-1]$ contains~$*$}~$\texttt{UPDATED-BEACONS}[u]\gets\texttt{empty}$\label{line:storeEmpty}
                \Else~$\texttt{UPDATED-BEACONS}[u]\gets(\bfy'[u,u+c\log m-1])$
                \EndIf
        \EndIf

    \EndFor
    \State $\texttt{BEACONS}\gets\textsc{repair-beacons}(\texttt{UPDATED-BEACONS},\texttt{R-STRINGS},\texttt{level})$\label{line:repairSigs}
    \ForAll{keys~$i$ in $\texttt{BEACONS}$}~$\bfy'[i:i+c\log m -1]\gets\texttt{BEACONS}[i]$ \EndFor
    \State $\texttt{UPDATED-BEACONS}\gets$ empty KV store.
    \State $\texttt{Z-FRAGMENTS},\bfy'\gets\textsc{anchor-fragments}(\texttt{BEACONS},\texttt{Z-FRAGMENTS},\bfy')$
\EndWhile
\Statex \textcolor{red}{$\triangleright~\textsc{recovery of residuals}$}
\State Let~\texttt{RESIDUALS} be an empty KV store.
\ForAll{keys~$i$ in $\texttt{BEACONS}$ in ascending order}\label{line:ForBeginResiduals}
    \State Let~$i_\text{next}$ be the smallest key greater than~$i$, or~$m+1$ if~$i$ is the greatest.
    \If{$\bfy'[i+c\log m:i_\text{next}-1]$ contains~$*$}~$\texttt{RESIDUALS}[i]\gets\texttt{empty}$
    \Else~$\texttt{RESIDUALS}[i]\gets\textsc{pad}(\bfy'[i+c\log m:i_\text{next}-1])$
    \EndIf
 \EndFor
\State $\texttt{REPAIRED-RESIDUALS}\gets \textsc{repair-residuals}(\texttt{RESIDUALS}, \texttt{R-STRINGS})$\label{line:repairResiduals}
\ForAll{keys~$i$ in $\texttt{REPAIRED-RESIDUALS}$}
    \State $\bfr\gets\textsc{de-pad}(\texttt{REPAIRED-RESIDUALS}[i]$)
    \State $\bfy'[i:i+\vert\bfr\vert -1]\gets\bfr$
\EndFor
\State \Return $\bfy'[c\log m+1:]$
\end{algorithmic}
\end{algorithm}

Having obtained~$\bfA$, the algorithm allocates string~$\bfy'=\bfm_0\circ*^{m}$ to represent the information region of~$\bfc$.
The~$*$'s represents the~$m$ unknown values of the original~$\bfz$, which are preceded by~$\bfm_0$.
Using the correct ordering of all level-$0$ beacons in~$\bfA$, as well as the distance between every two adjacent ones, the decoder traverses all adjacent pairs of level-$0$ beacons according to~$\bfA$ and positions all level-$0$ beacons appropriately in~$\bfy'$, making
\begin{equation*}
    \bfy'=\bfm_0\circ*^{|\bfz_1|}\circ\bfs_1\circ*^{|\bfz_2|}\circ\bfs_2\circ\cdots\circ\bfs_{r}\circ*^{|\bfz_{r+1}|}~\mbox{(line~\ref{line:allocateyp} of Alg.~\ref{alg:Decode1})},
\end{equation*}
where $\bfz_1,\bfz_2,\ldots,\bfz_{r+1}$ are intervals between level-$0$ beacons~$\bfs_1,\ldots,\bfs_r$, defined in Lemma~\ref{lemma:speratedIntervals}.
That is, the decoder has identified all level-$0$ beacons and their exact positions.

\subsubsection{\textbf{Recovery of higher-level beacons}}

We proceed to Algorithm~\ref{alg:Decode2}.
In line~\ref{line:decode-kv-beacons}, the decoder initializes a KV store \texttt{BEACONS}, where the values are the level-$0$ beacons and the keys are their corresponding indices in~$\bfy'$.
In line~\ref{line:decode-new-beacons}, it defines an empty KV store \texttt{UPDATED-BEACONS} to hold the higher-level beacons.
In line~\ref{line:affixz-fragments} which follows, the algorithm~\emph{anchors} the fragments in~$\texttt{Z-FRAGMENTS}$ to their correct position in~$\bfy'$; by~\emph{correct} we mean that~$\bfy'[i:i+\vert f\vert-1]=\bff$ if~$\bfy[i:i+\vert f\vert-1]=\bff$, where~$\bfy=\bfm_0\circ \bfz$ (line~\ref{line:attachm0}, Alg.~\ref{alg:Encoding}).
The correctness follows from Lemma~\ref{lemma:insideInfoRegion} and Property~\ref{item:distinct}.
The former assures that the anchored fragments from~\texttt{Z-FRAGMENTS} are all contained in the information region of the codeword.
The latter guarantees uniqueness of all substrings of length~$c\log m$ in~$\bfz$, and enables the use of level-$0$ beacons for synchronization purposes (line~\ref{line:sync1}--\ref{line:sync2}, Alg.~\ref{alg:utilities}).
Having anchored all such fragments, all remaining fragments in~\texttt{Z-FRAGMENTS} contain no level-$0$ beacons.

In the \textbf{while} loop starting at line \ref{line:WhileBegin}, the decoding algorithm proceeds with the extraction of higher level beacons.
This is done by repeatedly traversing all beacons that have already been anchored to~$\bfz$.
In traversal~$\ell$, the algorithm locates the midpoint~$u$ between every two adjacent beacons in~$\bfy'$, as long as the gap between them is large enough to fit at least one more beacon.
Then, the algorithm identifies and collects the $c\log m$ bits which begin at~$u$ as a level-$\ell$ beacon.
This will result in~\texttt{UPDATED-BEACONS}, a KV store of beacons from level-$0$ through~$\ell$.

Note that, some entries in~\texttt{UPDATED-BEACONS} may have the value~\texttt{empty}, in the case that the respective part of~$\bfy'$ contains a~$*$ (line~\ref{line:storeEmpty}, Alg.~\ref{alg:Decode2}).
However, as long as certain conditions hold, the number of such entries is bounded as follows.
\begin{lemma}\label{lemma:boundHigherLevelbeacons}
    Assume that in the beginning of the~\textbf{for} loop starting at line~\ref{line:ForBegin},
    \begin{enumerate}
        \item For every index~$i\in[m+c\log m]$, if~$\bfy'[i]\neq*$, then~$\bfy'[i]=\bfy[i]$, and\label{item:yAndyp}
        \item For every beacon~$\bfs=\bfy[i:i+c\log m-1]\in\cup_{j=0}^{\ell-1}\cS_j$ from level~$0$ through~$(\ell-1)$, it holds that
        $\texttt{BEACONS}[i]=\bfs$.
    \end{enumerate}
    Then, when the~\textbf{for} loop ends, it holds that~$\texttt{UPDATED-BEACONS}[i']=\bfs'$ for every beacon~$\bfs'=\bfy[i':i'+c\log m-1]\in\cup_{j=0}^{\ell}\cS_j$ from level~$0$ through~$\ell$, with at most~$2t_1$ exceptions.
\end{lemma}
\begin{proof}
    Given that all level-$0$ to level-$(\ell-1)$ beacons are stored in~\texttt{BEACONS} with their indices being the keys, the decoder is able to obtain the correct position of every level-$\ell$ beacon in~$\bfy$ (line~\ref{line:middlePoint}).
    Hence, due to Assumption~(\ref{item:yAndyp}),~$\texttt{UPDATED-BEACONS}[i]=\bfs$ for every beacon~$\bfs=\bfy[i:i+c\log m-1]$ from level~$0$ through~$\ell$, except for those that contain~$*$ in the respective interval in~$\bfy'$ (line~\ref{line:storeEmpty}).
    If the respective interval in~$\bfy'$ of a level-$\ell$ beacon contains~$*$'s, it must be (entirely or partially) contained in at least one fragment~$\bff$ that haven't been anchored yet.
    
    We refer such a fragment as an~\emph{level-$\ell$ unanchored fragment}, i.e., one that the decoder cannot anchor to~$\bfy'$ using beacons from levels~$0$ through~$\ell$, since none of these beacons is entirely contained within it.
    Note that such~$\bff$ must be entirely contained in the information region.
    Otherwise, since it contains bits from both redundancy region and~$\bfz$, it would have been split into two fragments (line~\ref{line:transpoint}).
    
    Observe that a fragment may be level-$\ell$ unanchored due to exactly one of the following reasons.
    \begin{enumerate}
        \item $\bff$ is not stored in~$\texttt{Z-FRAGMENTS}$ (i.e.,~$\vert\bff\vert<3c\log m$),  or
        \item $\bff$ is stored in~$\texttt{Z-FRAGMENTS}$, but does not contain any level-$0$ to level-$(\ell-1)$ beacons.
    \end{enumerate}
    Otherwise, such~$\bff$ would have been anchored to~$\bfz$.

    Notice that, a level-$\ell$ unanchored fragment may intersect with at most~$2$ level-$\ell$ beacons.
    Otherwise, i.e., if three level-$\ell$ beacons intersect with it, they must be separated by two beacons from level~$0$ through~$(\ell-1)$.
    Therefore,~$\bff$ contains lower level beacons and is of length~$\vert\bff\vert>3c\log m$.
    As a result, $\bff$ would have been contained in~\texttt{Z-FRAGMENTS} due to its length, and would have been anchored to~$\bfy'$ since it contains a lower level beacon, a contradiction.

    Let the information region of codeword~$\bfc$ be~$\bfc[i_\text{trans}:]$, where~$\bfc[i_\text{trans}:i_\text{trans}+c\log m-1]=\bfm_0$ and~$\bfc[i_\text{trans}+c\log m:]=\bfz$.
    Let~$j$ be the smallest index greater than~$i_\text{trans}$ such that~$c_j$ and~$c_{j+1}$ are contained in different fragments.
    If such an index~$j$ does not exist, then no break falls in the information region, and as a result,~$t_1=0$; in this case, no level-$\ell$ unanchored fragment exist.
    
    Otherwise, if such an index~$j$ does exist, let~$\bff_\text{trans}\triangleq\bfc[i:j]$ be the fragment containing~$c_j$, where~$i\leq i_\text{trans}$.
    The (at most)~$t_2$ fragments on the left of~$\bff_\text{trans}$ are not level-$\ell$ unanchored since they do not reside in the information region.
    We claim that~$\bff_\text{trans}$ is not level-$\ell$ unanchored as well:
    If~$\bff_\text{trans}$ contains~$\bfm_0$, then it would have been broken into two (line~\ref{line:transpoint}); the left one is clearly not level-$\ell$ unanchored (it does not contain bits from~$\bfz$), and so is the right one (it contains~$\bfm_0$, a level-$0$ beacon).
    If~$\bff_\text{trans}$ does not contain~$\bfm_0$, then the index~$j$ resides in~$\bfm_0$ and hence~$\bff_\text{trans}$  is not level-$\ell$ unanchored since it does not contain any bits from~$\bfz$.

    Meanwhile, each of the (at most)~$t_1$ fragments on the right of~$\bff_\text{trans}$ may be level-$\ell$ unanchored.
    Together they may intersect with at most $2t_1$ level~$\ell$ beacons, and as a result,~\texttt{UPDATED-BEACONS} has at most~$2t_1$ entries with value \texttt{empty}.
\end{proof}

The preceding lemma allows the decoder to repair~\texttt{UPDATED-BEACONS} and recover all level-$\ell$ beacons.
\begin{lemma}\label{lemma:RecoverAllbeaconsInALevel}
     Assume that for every level-$0$ to level-$\ell$ beacon that begins at index~$i$ in~$\bfy$,~$\texttt{UPDATED-BEACONS}[i]=\bfy[i:i+c\log m-1]$, with at most~$2t_1$ exceptions being~\texttt{empty}.
     Then, line~\ref{line:repairSigs} outputs a KV store~\texttt{BEACONS} such that for every level-$0$ to level-$\ell$ beacon that begins at index~$i$ in~$\bfy$,~$\texttt{BEACONS}[i]=\bfy[i:i+c\log m -1]$ (with no exceptions).
\end{lemma}
\begin{proof}
    In~\textsc{repair-beacons} (line~\ref{line:repairSigs}, Alg.~\ref{alg:utilities}), a codeword is constructed by coalescing the beacons stored in \texttt{UPDATED-BEACONS} with redundancy symbols stored in~\texttt{R-STRINGS} (line \ref{line:decodebeacons}, Alg.~\ref{alg:utilities}).
    Recall that the~$t-t_2$ redundancy strings contain~$2(t-t_2)$ redundancy symbols for each level of beacons, and as a result there exist~$2t_1$ erasures (\texttt{empty}) in the constructed codeword.
    The decoding in line~\ref{line:decodebeacons} of Alg.~\ref{alg:utilities} is guaranteed to be successful since
    $$
        2\cdot(t-t_2)-2\cdot t_1\geq 2\cdot (t-t_1-t_2)\geq 0,
    $$
    where the last transition follows since the actual number of breaks~$t_1+t_2$ is at most the security parameter~$t$.
    The proof is concluded since a~$(k+2t,k)$ Reed-Solomon code can correct~$2t$ erasures.
\end{proof}

The preceding lemmas enable the decoder to correctly anchor all fragments in~\texttt{Z-FRAGMENTS}, as seen in the following theorem.
\begin{theorem}
    The \textbf{while} loop starting at line \ref{line:WhileBegin} will eventually terminate, and by then every fragment in~\texttt{Z-FRAGMENTS} is correctly anchored to~$\bfy'$.
\end{theorem}
\begin{proof}
    Note that the assumptions in Lemma~\ref{lemma:boundHigherLevelbeacons} are true for level-$1$ beacons.
    Together with Lemma~\ref{lemma:RecoverAllbeaconsInALevel}, line~\ref{line:repairSigs} at the first iteration of the \textbf{while} loop outputs a KV store~\texttt{BEACONS} such that for every level-$0$ beacon that begins at index~$i$ in~$\bfy$,~$\texttt{BEACONS}[i]=\bfy[i:i+c\log m -1]$.
    It also allows correct placements of fragments containing a level-$1$ beacon in~\texttt{Z-FRAGMENTS}.

    By induction, beacons (in all levels), as well as the fragments containing them, will be correctly anchored to~$\bfy'$.
    Recall that by Lemma~\ref{lemma:atLeastOneSig}, every fragment in~\texttt{Z-FRAGMENTS} contains at least one beacon, and as a result, the \textbf{while} loop eventually terminates with every fragment in~\texttt{Z-FRAGMENTS} being correctly anchored to~$\bfy'$.
\end{proof}

\subsubsection{\textbf{Recovery of residuals}}
Finally, since all beacons and their containing fragments are anchored to~$\bfy'$, it follows that all~$*$'s in~$\bfy'$ are now contained in residual parts.
Similar to Lemma~\ref{lemma:boundHigherLevelbeacons}, we have the following theorem.
\begin{theorem}\label{theorem:RecoverAllResiduals}
    Line~\ref{line:repairResiduals} outputs a KV store~\texttt{REPAIRED-RESIDUALS} such that~$\texttt{REPAIRED-RESIDUALS}[i]=\textsc{pad}(\bfr)$ for every residual~$\bfr=\bfy[i:i+\vert\bfr\vert]$.
\end{theorem}
\begin{proof}
    We first show that the number of~\texttt{empty} entries in~\texttt{RESIDUALS} is bounded by~$3\cdot t_1$.
    If a residual contains~$*$'s, it must be (entirely or partially) contained in at least one fragment~$\bff$, referred to as an~\emph{unanchored fragment for residuals}, that has not been anchored to~$\bfy'$ yet.
    Note that such~$\bff$ must be entirely contained in the information region.
    Otherwise, since it contains bits from both the redundancy region and from~$\bfz$, it would have been split into two (line~\ref{line:transpoint}).

    Observe that a fragment may be unanchored-for-residuals due to exactly one of the following reasons.
    \begin{enumerate}
        \item $\bff$ is not stored in~$\texttt{Z-FRAGMENTS}$ (i.e.,~$\vert\bff\vert<3c\log m$),  or
        \item $\bff$ is stored in~$\texttt{Z-FRAGMENTS}$, but does not contain any beacons (in any level).
    \end{enumerate}
    Otherwise, such~$\bff$ would have been anchored to~$\bfz$.

    Notice that, an unanchored-for-residuals fragment~$\bff$ may intersect at most~$3$ residuals.
    Otherwise, i.e., if four beacons intersect with it and each of is length at least~$1$, they must be separated by three beacons.
    Therefore,~$\bff$ contains at least three beacons, and therefore it is of length at least~$3c\log m$.
    As a result, $\bff$ would have resided in~\texttt{Z-FRAGMENTS} due to its length, and would have been anchored to~$\bfy'$ since it contains a lower level beacon, a contradiction.

    Let the information region of the codeword~$\bfc$ be~$\bfc[i_\text{trans}:]$, where~$\bfc[i_\text{trans}:i_\text{trans}+c\log m-1]=\bfm_0$ and~$\bfc[i_\text{trans}+c\log m:]=\bfz$.
    Let~$j$ be the smallest index greater than~$i_\text{trans}$ such that~$c_j$ and~$c_{j+1}$ are contained in different fragments.
    If such an index does not exist, then no break falls in the information region, and as a result,~$t_1=0$; in this case, no unanchored fragments exist.
    
    Otherwise, if such an index does exist, let~$\bff_\text{trans}\triangleq\bfc[i:j]$ be the fragment containing~$c_j$, where~$i\leq i_\text{trans}$.
    The (at most)~$t_2$ fragments on the left of~$\bff_\text{trans}$ are not unanchored since they are not contained in the information region.
    We claim that~$\bff_\text{trans}$ is not unanchored as well:
    If~$\bff_\text{trans}$ contains~$\bfm_0$, then it would have been broken into two (line~\ref{line:transpoint}); the left one is clearly not unanchored-for-residual (it does not contain bits from~$\bfz$), and so is the right one (it contains~$\bfm_0$, a level-$0$ beacon).
    If~$\bff_\text{trans}$ does not contain~$\bfm_0$, then the index~$j$ resides in~$\bfm_0$ and hence~$\bff_\text{trans}$  is not unanchored since it does not contain any bits from~$\bfz$.

    Meanwhile, each of the (at most)~$t_1$ fragments on the right of~$\bff_\text{trans}$ may be unanchored-for-residual.
    Together, they may intersect with at most $3\cdot t_1$ level~$\ell$ residuals, and as a result,~\texttt{UPDATED-BEACONS} has at most~$3\cdot t_1$ entries with value \texttt{empty}.    
\end{proof}

Finally, the decoder anchors all residuals in~\texttt{REPAIRED-RESIDUALS} to~$\bfz$, and after which~$\bfy'=\bfy$.
Recall that the marker~$\bfm_0$ is attached to the left of~$\bfz$ by the decoder, and needs to be removed to obtain~$\bfy$.
This implies the following, which concludes the proof of correctness of our construction.
\begin{theorem}
    Let~$\bfz\in\bi^m$ be a legit string, let~$\bfc$ be the output of Algorithm~\ref{alg:Encoding} with input~$\bfz$, and let~$\bff_1,\ldots,\bff_{\ell}$ be fragments of~$\bfc$ for some~$\ell\le t+1$. Then, the decoding algorithm with inputs~$\bff_1,\ldots,\bff_\ell$ outputs~$\bfz$.
\end{theorem}

In real-word scenarios, it may be crucial in security-critical applications to tolerate fragment losses, as the adversary may choose to hide a portion of the bits in order to fail the decoding.
Although this is not our main purpose in this paper, our codes have the added benefit of tolerating losses of short fragments, as shown in the following theorem.

\begin{theorem}
    The proposed~$(n, t)$-BRC is tolerant to the loss of any of the~$t+1$ fragments whose total length is less than~$c\log m$ bits.
\end{theorem}
\begin{proof}
    The decoding algorithm uses fragments in~\texttt{R-FRAGMENTS} and~\texttt{Z-FRAGMENTS}, and both of them exclude fragments that are shorter than~$c\log m$ bits (line~\ref{line:R-FRAGMENTS}--\ref{line:Z-FRAGMENTS}, Alg.~\ref{alg:Decode1}).
\end{proof}

\section{Redundancy Analysis}\label{section:redundancy}
We now present an analysis of the redundancy in the encoding process (Section~\ref{section:encoding}), whose crux is bounding the success probability of choosing a legit binary string~$\bfz$ (Definition~\ref{def:legit}).
Additional components of the redundancy, such as Reed-Solomon parity symbols and markers, are easier to analyze and will be addressed in the sequel.
For the following theorem, recall that~$m=|\bfz|$ and~$\cC_\text{MU}$ is a mutually uncorrelated code of length~$c\log m$ and size~$\vert\cC_\text{MU}\vert\geq\frac{2^{c\log m}}{\beta c\log m}$, where~$\beta= 32$.
The following theorem forms the basis of our analysis; it is inspired by ideas from~\cite[Theorem~4.4]{cheng2018deterministic}. 
\begin{theorem}\label{theorem:legit-probability}
    A uniformly random string $\bfz\in\bi^m$ is legit (Definition~\ref{def:legit}) with probability~$1-1/\operatorname{poly}(m)$.
\end{theorem}
\begin{proof}
    For Property~\ref{item:distance}, let~$I$ be an arbitrary interval of~$\bfz$ of length~$d+c\log m-1$, where~$d=2\beta c\log^2m$.
    For~$i\in[d]$, let~$E_i$ be the event that~$I[i:i+c\log m -1]$ is a level-$0$ beacon (i.e., a codeword of~$\cC_{\text{MU}}$), let~$S_i=\cup_{j=1}^i E_i$,
    and let~$E_i^c$ and~$S_i^c$ be their complements. 
    We bound the probability of Property~\ref{item:distance} by first bounding~$\Pr(S_d^c)$, i.e., the probability that no level-$0$ beacon starts at the first~$d$ bits of~$I$, which is equivalent to the probability that~$I$ contains a level-$0$ beacon.
    Then, we apply the union bound over all such intervals of~$\bfz$.

    For distinct~$i,j\in[d]$ we have the following observations.
    \begin{enumerate}[start=1,label={(\Alph*)}]
        \item If~$j-i< c\log m$ then~$E_i\cap E_j=\varnothing$, since codewords of~$\cC_\text{MU}$ cannot overlap. \label{item:emptyintersection}
        \item If~$j-i\ge c\log m$ then~$E_i$ and~$E_j$ are independent, since~$E_i$ depends on bits in~$I[i:i+c\log m-1]$ and~$E_j$ depends on bits in~$I[j:j+c\log m-1]$.\label{item:independent}   
    \end{enumerate}

    Since~$S_d^c=E_d^c\cap\ldots\cap E_{c\log m+1}^c\cap S_{c\log m}^c$, the chain rule for probability implies that
    \begin{equation}\label{eq:probablity-cw-free}
        \Pr(S_d^c)=\Pr(E_d^c\vert \cap_{i=1}^{d-1} E_i^c)\Pr(E_{d-1}^c\vert \cap_{i=1}^{d-2}E_i^c)\cdot\ldots \cdot\Pr(E_{c\log m+1}^c\vert S_{c\log m}^c)\Pr(S_{c\log m}^c).
    \end{equation}
    To bound the rightmost term in~\eqref{eq:probablity-cw-free},
    \begin{equation}\label{eq:first-clogm-indices}
        \begin{split}
            \Pr(S_{c\log m}^c)&=1-\Pr(\cup_{i=1}^{c\log m}E_i)
            \overset{(a)}{=}1-\sum_{i}\Pr(E_i)+\sum_{i,j}\Pr(E_i\cap E_j)-\sum_{i,j,k}\Pr(E_i\cap E_j\cap E_k)+\ldots\\
            &\overset{(b)}{=}1-\sum_{i\in[c\log m]}\Pr(E_i) = 1- c\log m \cdot P,~\mbox{where }P\triangleq\Pr(E_i)=\frac{\vert\cC_\text{MU}\vert}{2^{c\log m}}\overset{(c)}{\geq} \frac{1}{\beta c\log m}.
        \end{split}
    \end{equation}
    Note that~$(a)$ follows from the inclusion-exclusion principle for probability events, $(b)$ follows from Observation~\ref{item:emptyintersection} above, and~$(c)$ follows from the fact that~$\vert\cC_\text{MU}\vert\geq\frac{2^{c\log m}}{\beta c\log m}$.

    To bound the remaining terms in~\eqref{eq:probablity-cw-free}, it follows from Observation~\ref{item:independent} that 
    \begin{align}\label{equation:dependsOnPrevious}
        P_s&\triangleq \Pr(E_j^c\vert \cap_{i=1}^{j-1}E_i^c)=\Pr(E_j^c \vert \cap_{i=j-c\log m+1}^{j-1} E_i^c)
    \end{align}
    for every~$j\in[c\log m+1:d]$, and furthermore, \eqref{equation:dependsOnPrevious} is identical for every such~$j$.
    Notice that~$\Pr(S_{c\log m}^c)=P_s\cdot \Pr(S_{c\log m-1}^c)$, and as a result,
    \begin{equation}\label{eq:conditional-probability}
         P_s=\frac{\Pr(S_{c\log m}^c)}{\Pr(S_{c\log m-1}^c)}=\frac{1-P\cdot c\log m}{1-P\cdot (c\log m -1)}=1-\frac{P}{1-P\cdot (c\log m -1)}\leq 1-P,
    \end{equation}
    where the last inequality follows since~$1-P\cdot(c\log m-1)$ is positive and at most~$1$.
    To bound~$\Pr(S_{d}^c)$, we combine the above as follows.
    \begin{equation}
        \begin{split}
             \Pr(S_{d}^c) & \overset{\eqref{eq:probablity-cw-free}}{=} \left[\textstyle\prod_{j=c\log m+1}^d \Pr(E_{j}^c\vert\cap_{i=1}^{j-1}E_i^c)\right] \cdot \Pr(S_{c\log m}^c)\overset{\eqref{equation:dependsOnPrevious}}{=} \left[\textstyle\prod_{j=c\log m+1}^d \Pr(E_{j}^c\vert\cap_{i=j-c\log m+1}^{j-1}E_i^c)\right] \cdot \Pr(S_{c\log m}^c)\\
             &=P_s^{d-c\log m }\cdot \Pr(S_{c\log m}^c)\overset{\eqref{eq:conditional-probability}}{\le} (1-P)^{d-c\log m}\cdot  \Pr(S_{c\log m}^c)
             \overset{\eqref{eq:first-clogm-indices}}{=} (1-P)^{d-c\log m}\cdot (1-c\log m\cdot P)\\
             &\overset{(d)}{\leq} (1-P)^{d-c\log m}\cdot  (1- P)^{c\log m}
             = \left(1-\textstyle\frac{1}{\beta c\log m}\right)^d\\
             &=\left(1-\textstyle\frac{1}{\beta c\log m}\right)^{2\beta c\log^2m}\leq e^{-2\log m}\leq m^{-2},
        \end{split}
    \end{equation}
    where $(d)$ is known as Bernoulli's inequality\footnote{Bernoulli's inequality states that~$1+rx\le (1+x)^r$ for every real number~$r\ge1$ and~$x\ge-1$.}.
    Therefore, the probability of an interval of length~$2\beta c\log^2m+c\log m-1$ to be free of codewords of~$\cC_\text{MU}$ is bounded by~$m^{-2}$.
    By applying the union bound, we have that the probability that such an interval exists in~$\bfz$ is at most~$m^{-1}$. Therefore, the probability that Property~\ref{item:distance} holds, i.e., that every interval of~$\bfz$ of length~$2\beta c\log^2m+c\log m -1$ contains at least one codeword of~$\cC_\text{MU}$, is at least~$1-1/m$.

    For property~\ref{item:distinct}, consider any two intervals of length~$d=c\log m$ starting at indices~$i,j$ such that~$j-i\geq d$.
    Then,
    \begin{align*}    
    \Pr(\bfz[i:i+d-1]=\bfz[j:j+d -1))=\sum_{r\in\bi^{d}}\Pr(\bfz[j:j+d-1]=r\mid \bfz[i:i+d-1]=r)\Pr(\bfz[i,i+d-1]=r)=\frac{1}{2^{d}}.
    \end{align*}
    By applying the union bound, the probability that such~$i,j$ exist is at most~$\frac{m^2}{2^d}=\frac{m^2}{2^{c\log m}} =m^{-c+2}$.
 
    For Property~\ref{item:no-markers}, note that the probability that an interval of length~$c\log m$ of~$\bfz$ matches one of the markers is~$(t+1)m^{-c}$.
    By applying the union bound, the probability that such an interval exists is at most
    \begin{equation*}
        \frac{(t+1)(m-c\log m+1)}{m^c}<\frac{(t+1)m}{m^c}<m^{-c+2}.    
    \end{equation*}

    To conclude, a uniformly random~$\bfz\in\bi^m$ does not satisfy either of the three properties with probability~$1/\operatorname{poly}(m)$ each. Therefore,~$\bfz$ is legit, i.e., satisfies all three properties, with probability at least~$1-1/\text{poly(m)}$ by the union bound. 
\end{proof}

With Theorem~\ref{theorem:legit-probability}, we can formally provide the redundancy of our scheme in the following corollary.
\begin{corollary}
    The code has redundancy of~$O(t\log n\log\log n)$.
\end{corollary}
\begin{proof}
    As shown in Section~\ref{section:code}, the codeword length~$n=\vert\bfc\vert = m+(6+2\log\log m)\cdot 3c\log m\cdot t +c\log m$.
    Meanwhile, the code size~$\vert\cC\vert$ depends on the probability of rejecting a uniform random string in~$\bi^m$ during code construction.
    By Theorem~\ref{theorem:legit-probability}, 
    $$
    \log \vert\cC\vert > \log [2^m\cdot (1-1/\text{poly(m)})]=m + \log (1-1/\text{poly(m)})=m-o(1).
    $$
    As a result, the redundancy is
    $$
        n-\log|\cC|=(6+2\log\log m)\cdot 3c\log m\cdot t +c\log m+o(1)=O(t\log m\log\log m).
    $$
    Note that~$n>m$, the redundancy is then~$O(t\log n\log\log n)$. 
\end{proof}

\section{Discussion and Future Work}\label{section:discussion}
In this paper, we analyze a new adversarial noise model in which an adversary can arbitrarily break the transmitted information.
We introduce $(n,t)$-break-resilient codes, which is a family of length-$n$ codes that guarantee correct reconstruction under any pattern of up to $t$ breaks.
We prove the existence of such codes with redundancy $O(t \log n)$ and show that this bound is indeed information-theoretically optimal.

We further present an explicit code construction that achieves redundancy $O(t \log n \log \log n)$.
A key idea in our design is leveraging naturally occurring patterns in a uniformly random string (Definition~\ref{def:legit}) as beacons, which reduces redundancy.
While inserting markers at variable positions is relatively straightforward, ensuring that these markers simultaneously satisfy the additional properties of Definition~\ref{def:legit} is nontrivial.
In particular, although our probabilistic analysis implies the existence of an injective mapping from $\bi^{m-1}$ to $\bi^m$ whose output is always legit, constructing such a mapping deterministically remains challenging and is left as an open problem.


\bibliographystyle{IEEEtran}
\bibliography{ref.bib} 

\begin{appendices}

\section{Histogram-based code construction over large alphabets}\label{appendix:histogram}
Let~$\Sigma$ be an alphabet with~$q$ symbols for some positive integer~$q$.
For~$a,b\in\Sigma^n$, we say $a$ is equivalent to~$b$ if they are form by the same multiset of symbols in~$\Sigma$, i.e., they share the same histogram of symbols.
A histogram-based code is formed by choosing exactly one word from every equivalence class, and hence no two codeword share the same histogram.

To decode the histogram-based code, the decoder simply counts the occurrence of symbols in the fragments and learns which equivalence class the codeword belongs to.
Since no two codewords share the same histogram, the decoding is guaranteed to be successful.

Clearly, the number of codewords equals to the number of such equivalence classes.
Using the stars and bars formula, this number is
\begin{equation}\label{eq:num-of-eq-class}
    \binom{q+n-1}{n}=\frac{(q+n-1)!}{n!(q-1)!}=\frac{(q+n-1)(q+n-2)\cdots q}{n!}.
\end{equation}
The redundancy of a histogram-based code is therefore
\begin{equation}
 n-\log_q\left(\frac{(q+n-1)(q+n-2)\cdots q}{n!}\right)=
 \log_q\left(n!\frac{q}{q+n-1}\frac{q}{q+n-2}\cdots\frac{q}{q}\right)<\log_q (n!).
\end{equation}
Hence, if~$q^c\geq n!$ for some constant~$c$, then the redundancy is less than~$c$ symbols.

\section{Auxiliary Functions}\label{appendix:auxiliary}
\begin{algorithm}\caption{Auxiliary Functions}\label{alg:utilities}
\begin{algorithmic}[1]
\Function{\textsc{rs-encode}}{$(\texttt{info-word}),\texttt{num-parity}$}\label{line:encodeRS}
\EndFunction
\Function{\textsc{rs-decode}}{$(\texttt{codeword})$}\label{line:decodeRS}
\EndFunction

\Function{\textsc{compress-row}}{$\bfa$}\label{line:compressRow}
\EndFunction
\Function{\textsc{decompress-row}}{$\texttt{compRow}$}\label{line:decompressRow}
\EndFunction

\Function{\textsc{compress-adjacency-matrix}}{$\bfA$}\label{line:compress-adjacency-matrix}
\State Let~\texttt{COMP-A} be an empty array.
    \ForAll{row~$\bfa$ in the matrix~$\bfA$}
        \State $\texttt{COMP-A}.\textsc{append}(\textsc{compress-row}(\bfa))$\label{line:compressIndividualRow}
    \EndFor
\State \textbf{return}~$\texttt{COMP-A}$
\EndFunction

\Function{\textsc{repair-adj-matrix}}{$\texttt{APPROX-COMP-A},\texttt{R-STRINGS}$}\label{line:func-repair-adj-matrix}
    \State $
        \begin{aligned}
           \texttt{rs-decoded}\gets\textsc{rs-decode}(&\texttt{APPROX-COMP-A}[1],\ldots,\texttt{APPROX-COMP-A}[\vert \cC_\text{MU}\vert],\\
           &\texttt{R-STRINGS}[1][1:2c\log m],\ldots,\texttt{R-STRINGS}[1][6c\log m+1:8c\log m]\\
           &\cdots\\
           &\texttt{R-STRINGS}[t][1:2c\log m],\ldots,\texttt{R-STRINGS}[t][6c\log m+1:8c\log m]).\\
        \end{aligned}$\label{line:decodeAdjMatrix}
    \State Let~$\bfA$ be a matrix such that for every~$a\in[\vert \cC_\text{MU}\vert]$, its~$a$-th row is~$\textsc{decompress-row}(\texttt{decoded}[a])$.
    \State \Return $\bfA$
\EndFunction

\Function{\textsc{repair-beacons}}{$\texttt{BEACONS},\texttt{R-STRINGS},\texttt{level}$}\label{line:func-repair-beacons}
    \State Let~$i_1,\ldots,i_r$ be the keys in \texttt{BEACONS}
    \State $v\gets 8c\log m+(\texttt{level}-1)\cdot c\log m$
    \State $
            \begin{aligned}
                \texttt{decoded}\gets\textsc{rs-decode}(&\texttt{BEACONS}[i_1],\ldots,\texttt{BEACONS}[i_r],\\
                &\texttt{R-STRINGS}[0][v+1:v+c\log m],\texttt{R-STRINGS}[0][v+c\log m+1:v+2c\log m+1]\\
                &\cdots\\
                &\texttt{R-STRINGS}[t][v+1:v+c\log m],\texttt{R-STRINGS}[t][v+c\log m+1:v+2c\log m+1]).\\
            \end{aligned}$\label{line:decodebeacons}
    \ForAll{$j\in[r]$}~$\texttt{REPAIRED-BEACONS}\gets\texttt{decoded}[j]$\EndFor
\EndFunction

\Function{\textsc{repair-residuals}}{$\texttt{RESIDUALS},\texttt{R-STRINGS}$}\label{line:func-repair-residuals}
    \State Let~$i_1,\ldots,i_r$ be the keys in \texttt{RESIDUALS}
    \State $v\gets (4+2\log\log m)\cdot c\log m$
    \State $
            \begin{aligned}
                \texttt{decoded}\gets\textsc{rs-decode}(&\texttt{RESIDUALS}[i_1],\ldots,\texttt{RESIDUALS}[i_r],\\
                &\texttt{R-STRINGS}[0][v+1:v+c\log m],\ldots,\texttt{R-STRINGS}[0][v+2c\log m+1:v+3c\log m]\\
                &\cdots\\
                &\texttt{R-STRINGS}[t][v+1:v+c\log m],\ldots,\texttt{R-STRINGS}[t][v+2c\log m+1:v+3c\log m]).\\
            \end{aligned}$\label{line:decodeResiduals}
    \ForAll{$j\in[r]$}~$\texttt{REPAIRED-RESIDUALS}\gets\texttt{decoded}[j]$\EndFor
\EndFunction

\Function{\textsc{anchor-fragments}}{$\texttt{BEACONS},\texttt{Z-FRAGMENTS},\bfz$}
    \ForAll{$\bff$ in \texttt{Z-FRAGMENTS} and beacons~$\bfs$ in~\texttt{BEACONS}}.
        \If{there exist~$i,j$ such that~$\bfs=\bff[i:i+c\log m-1]=\texttt{BEACONS}[j]$}\label{line:sync1}
            \State $\bfz[j-i:j-i+|\bff|-1]\gets\bff$\label{line:sync2}
            \State \textbf{Delete}~$\bff$ from~\texttt{Z-FRAGMENTS}
        \EndIf
    \EndFor
    \State \Return \texttt{Z-FRAGMENTS},~$\bfz$
\EndFunction

\Function{\textsc{pad}}{$\bfs$}\label{line:paddingFunction}
    \State $\bfs\gets\bfs\circ 1$
    \While{$\vert\bfs\vert<c\log m$} $\bfs\gets\bfs\circ0$\EndWhile
    \State \Return $\bfs$
\EndFunction

\Function{\textsc{de-pad}}{$\bfs$}
    \While{$\bfs[1]=0$} remove the first bit from~$\bfs$\EndWhile
\State \Return $\bfs[1:\vert\bfs\vert-1]$
\EndFunction

\end{algorithmic}
\end{algorithm}

\end{appendices}
 
\end{document}

%% file: header.tex
\usepackage{enumerate}
\usepackage[shortlabels]{enumitem}

\usepackage{amssymb}
\usepackage{amsfonts}
\usepackage{mathrsfs}
\usepackage{amsmath} 
\usepackage{amsthm}
\usepackage{mathtools}

\newcommand{\ceil}[1]{ {\lceil#1\rceil}}

\newcommand{\bi}{{\{0,1\}}}

\usepackage{algorithm}
\usepackage[noend]{algpseudocode}
\algrenewcommand\algorithmiccomment[1]{\textcolor{lightgray}{\hfill // #1}}
\newcommand{\Commentx}[1]{\Statex\textcolor{lightgray}{//~#1}}

\usepackage{pgfplots}
\pgfplotsset{compat=1.18}
\usepackage{graphicx}

\usepackage{hyperref}
\hypersetup{
 colorlinks,
 linkcolor={blue!100!black},
 citecolor={blue!100!black},
 urlcolor={blue!80!black} 
}

 \usepackage{cite}
\usepackage{subcaption}
\usepackage{wrapfig}

\newtheorem{theorem}{Theorem}[section] 
\newtheorem{lemma}[theorem]{Lemma}
\newtheorem{corollary}[theorem]{Corollary}

\theoremstyle{definition}
\newtheorem{definition}[theorem]{Definition}

\newtheorem{remark}{Remark}
 \newtheorem{example}{Example}

\usepackage{xcolor}
\usepackage{soul} 

\newcommand{\bbF}{\mathbb{F}}

\newcommand{\bbN}{\mathbb{N}}

\newcommand{\bfa}{\mathbf{a}}

\newcommand{\bfc}{\mathbf{c}}
\newcommand{\bfd}{\mathbf{d}}

\newcommand{\bff}{\mathbf{f}}

\newcommand{\bfm}{\mathbf{m}}

\newcommand{\bfr}{\mathbf{r}}
\newcommand{\bfs}{\mathbf{s}}
\newcommand{\bft}{\mathbf{t}}
\newcommand{\bfu}{\mathbf{u}}

\newcommand{\bfx}{\mathbf{x}}
\newcommand{\bfy}{\mathbf{y}}
\newcommand{\bfz}{\mathbf{z}}

\newcommand{\bfA}{\mathbf{A}}

\newcommand{\cC}{\mathcal{C}}

\newcommand{\cS}{\mathcal{S}}

\newcommand{\cX}{\mathcal{X}}
\newcommand{\cY}{\mathcal{Y}}


%% file: figures/instrumentation.tex
\tikzset{every picture/.style={line width=0.75pt}} 

\begin{tikzpicture}[x=0.75pt,y=0.75pt,yscale=-1,xscale=1]

\draw   (11,56) -- (163,56) -- (163,80) -- (11,80) -- cycle ;
\draw    (171,72) -- (288,72) ;
\draw [shift={(291,72)}, rotate = 180] [fill={rgb, 255:red, 0; green, 0; blue, 0 }  ][line width=0.08]  [draw opacity=0] (8.93,-4.29) -- (0,0) -- (8.93,4.29) -- cycle    ;
\draw   (380,56) -- (420,56) -- (420,80) -- (380,80) -- cycle ;
\draw   (300,56) -- (380,56) -- (380,80) -- (300,80) -- cycle ;
\draw    (300,88) -- (380,88) ;
\draw [shift={(380,88)}, rotate = 180] [color={rgb, 255:red, 0; green, 0; blue, 0 }  ][line width=0.75]    (0,3.35) -- (0,-3.35)(6.56,-1.97) .. controls (4.17,-0.84) and (1.99,-0.18) .. (0,0) .. controls (1.99,0.18) and (4.17,0.84) .. (6.56,1.97)   ;
\draw [shift={(300,88)}, rotate = 0] [color={rgb, 255:red, 0; green, 0; blue, 0 }  ][line width=0.75]    (0,3.35) -- (0,-3.35)(6.56,-1.97) .. controls (4.17,-0.84) and (1.99,-0.18) .. (0,0) .. controls (1.99,0.18) and (4.17,0.84) .. (6.56,1.97)   ;
\draw    (380,88) -- (420,88) ;
\draw [shift={(420,88)}, rotate = 180] [color={rgb, 255:red, 0; green, 0; blue, 0 }  ][line width=0.75]    (0,3.35) -- (0,-3.35)(6.56,-1.97) .. controls (4.17,-0.84) and (1.99,-0.18) .. (0,0) .. controls (1.99,0.18) and (4.17,0.84) .. (6.56,1.97)   ;
\draw [shift={(380,88)}, rotate = 0] [color={rgb, 255:red, 0; green, 0; blue, 0 }  ][line width=0.75]    (0,3.35) -- (0,-3.35)(6.56,-1.97) .. controls (4.17,-0.84) and (1.99,-0.18) .. (0,0) .. controls (1.99,0.18) and (4.17,0.84) .. (6.56,1.97)   ;
\draw   (500,56) -- (540,56) -- (540,80) -- (500,80) -- cycle ;
\draw   (420,56) -- (500,56) -- (500,80) -- (420,80) -- cycle ;
\draw    (420,88) -- (500,88) ;
\draw [shift={(500,88)}, rotate = 180] [color={rgb, 255:red, 0; green, 0; blue, 0 }  ][line width=0.75]    (0,3.35) -- (0,-3.35)(6.56,-1.97) .. controls (4.17,-0.84) and (1.99,-0.18) .. (0,0) .. controls (1.99,0.18) and (4.17,0.84) .. (6.56,1.97)   ;
\draw [shift={(420,88)}, rotate = 0] [color={rgb, 255:red, 0; green, 0; blue, 0 }  ][line width=0.75]    (0,3.35) -- (0,-3.35)(6.56,-1.97) .. controls (4.17,-0.84) and (1.99,-0.18) .. (0,0) .. controls (1.99,0.18) and (4.17,0.84) .. (6.56,1.97)   ;
\draw    (500,88) -- (540,88) ;
\draw [shift={(540,88)}, rotate = 180] [color={rgb, 255:red, 0; green, 0; blue, 0 }  ][line width=0.75]    (0,3.35) -- (0,-3.35)(6.56,-1.97) .. controls (4.17,-0.84) and (1.99,-0.18) .. (0,0) .. controls (1.99,0.18) and (4.17,0.84) .. (6.56,1.97)   ;
\draw [shift={(500,88)}, rotate = 0] [color={rgb, 255:red, 0; green, 0; blue, 0 }  ][line width=0.75]    (0,3.35) -- (0,-3.35)(6.56,-1.97) .. controls (4.17,-0.84) and (1.99,-0.18) .. (0,0) .. controls (1.99,0.18) and (4.17,0.84) .. (6.56,1.97)   ;
\draw   (620,56) -- (660,56) -- (660,80) -- (620,80) -- cycle ;
\draw   (540,56) -- (620,56) -- (620,80) -- (540,80) -- cycle ;
\draw    (540,88) -- (620,88) ;
\draw [shift={(620,88)}, rotate = 180] [color={rgb, 255:red, 0; green, 0; blue, 0 }  ][line width=0.75]    (0,3.35) -- (0,-3.35)(6.56,-1.97) .. controls (4.17,-0.84) and (1.99,-0.18) .. (0,0) .. controls (1.99,0.18) and (4.17,0.84) .. (6.56,1.97)   ;
\draw [shift={(540,88)}, rotate = 0] [color={rgb, 255:red, 0; green, 0; blue, 0 }  ][line width=0.75]    (0,3.35) -- (0,-3.35)(6.56,-1.97) .. controls (4.17,-0.84) and (1.99,-0.18) .. (0,0) .. controls (1.99,0.18) and (4.17,0.84) .. (6.56,1.97)   ;
\draw    (620,88) -- (660,88) ;
\draw [shift={(660,88)}, rotate = 180] [color={rgb, 255:red, 0; green, 0; blue, 0 }  ][line width=0.75]    (0,3.35) -- (0,-3.35)(6.56,-1.97) .. controls (4.17,-0.84) and (1.99,-0.18) .. (0,0) .. controls (1.99,0.18) and (4.17,0.84) .. (6.56,1.97)   ;
\draw [shift={(620,88)}, rotate = 0] [color={rgb, 255:red, 0; green, 0; blue, 0 }  ][line width=0.75]    (0,3.35) -- (0,-3.35)(6.56,-1.97) .. controls (4.17,-0.84) and (1.99,-0.18) .. (0,0) .. controls (1.99,0.18) and (4.17,0.84) .. (6.56,1.97)   ;
\draw   (660,56) -- (684,56) -- (684,80) -- (660,80) -- cycle ;

\draw (233,61.5) node   [align=left] {\begin{minipage}[lt]{74.8pt}\setlength\topsep{0pt}
Instrumentation
\end{minipage}};
\draw (339.5,68) node   [align=left] {\begin{minipage}[lt]{55.08pt}\setlength\topsep{0pt}
\begin{center}
$\displaystyle \mathbf{m}_{l}$
\end{center}

\end{minipage}};
\draw (92,68) node   [align=left] {\begin{minipage}[lt]{103.36pt}\setlength\topsep{0pt}
$\displaystyle \mathbf{u}_{l} =\mathbf{u}_{l,1} \circ $$\displaystyle \mathbf{u}_{l,2} \circ $$\displaystyle \mathbf{u}_{l,3} \circ \cdots $
\end{minipage}};
\draw (340,104) node   [align=left] {\begin{minipage}[lt]{54.4pt}\setlength\topsep{0pt}
\begin{center}
{\scriptsize $\displaystyle c\log m$}
\end{center}

\end{minipage}};
\draw (399,104) node   [align=left] {\begin{minipage}[lt]{27.2pt}\setlength\topsep{0pt}
\begin{center}
{\scriptsize $\displaystyle \frac{c\log m}{2}$}
\end{center}

\end{minipage}};
\draw (672,68) node   [align=left] {\begin{minipage}[lt]{21.76pt}\setlength\topsep{0pt}
\begin{center}
$\displaystyle \cdots $
\end{center}

\end{minipage}};
\draw (520,68) node   [align=left] {\begin{minipage}[lt]{27.2pt}\setlength\topsep{0pt}
\begin{center}
$\displaystyle \mathbf{u}_{l,1}$
\end{center}

\end{minipage}};
\draw (460,68) node   [align=left] {\begin{minipage}[lt]{54.4pt}\setlength\topsep{0pt}
\begin{center}
$\displaystyle \mathbf{m}_{l}$
\end{center}

\end{minipage}};
\draw (460,104) node   [align=left] {\begin{minipage}[lt]{54.4pt}\setlength\topsep{0pt}
\begin{center}
{\scriptsize $\displaystyle c\log m$}
\end{center}

\end{minipage}};
\draw (519,104) node   [align=left] {\begin{minipage}[lt]{27.2pt}\setlength\topsep{0pt}
\begin{center}
{\scriptsize $\displaystyle \frac{c\log m}{2}$}
\end{center}

\end{minipage}};
\draw (640,68) node   [align=left] {\begin{minipage}[lt]{27.2pt}\setlength\topsep{0pt}
\begin{center}
$\displaystyle \mathbf{u}_{l,1}$
\end{center}

\end{minipage}};
\draw (580,68) node   [align=left] {\begin{minipage}[lt]{54.4pt}\setlength\topsep{0pt}
\begin{center}
$\displaystyle \mathbf{m}_{l}$
\end{center}

\end{minipage}};
\draw (579,104) node   [align=left] {\begin{minipage}[lt]{54.4pt}\setlength\topsep{0pt}
\begin{center}
{\scriptsize $\displaystyle c\log m$}
\end{center}

\end{minipage}};
\draw (639,104) node   [align=left] {\begin{minipage}[lt]{27.2pt}\setlength\topsep{0pt}
\begin{center}
{\scriptsize $\displaystyle \frac{c\log m}{2}$}
\end{center}

\end{minipage}};
\draw (400,68) node   [align=left] {\begin{minipage}[lt]{27.2pt}\setlength\topsep{0pt}
\begin{center}
$\displaystyle \mathbf{u}_{l,1}$
\end{center}

\end{minipage}};

\end{tikzpicture}

%% file: main.bbl
\begin{thebibliography}{10}
\providecommand{\url}[1]{#1}
\csname url@samestyle\endcsname
\providecommand{\newblock}{\relax}
\providecommand{\bibinfo}[2]{#2}
\providecommand{\BIBentrySTDinterwordspacing}{\spaceskip=0pt\relax}
\providecommand{\BIBentryALTinterwordstretchfactor}{4}
\providecommand{\BIBentryALTinterwordspacing}{\spaceskip=\fontdimen2\font plus
\BIBentryALTinterwordstretchfactor\fontdimen3\font minus \fontdimen4\font\relax}
\providecommand{\BIBforeignlanguage}[2]{{%
\expandafter\ifx\csname l@#1\endcsname\relax
\typeout{** WARNING: IEEEtran.bst: No hyphenation pattern has been}%
\typeout{** loaded for the language `#1'. Using the pattern for}%
\typeout{** the default language instead.}%
\else
\language=\csname l@#1\endcsname
\fi
#2}}
\providecommand{\BIBdecl}{\relax}
\BIBdecl

\bibitem{elsayed2021information}
K.~A. ElSayed, A.~Dachowicz, and J.~H. Panchal, ``Information embedding in additive manufacturing through printing speed control,'' in \emph{Proceedings of the 2021 Workshop on Additive Manufacturing (3D Printing) Security}, 2021, pp. 31--37.

\bibitem{voris2017three}
J.~Voris, B.~F. Christen, J.~Alted, and D.~W. Crawford, ``Three dimensional (3d) printed objects with embedded identification (id) elements,'' May 2017, uS Patent 9,656,428.

\bibitem{wei2018embedding}
C.~Wei, Z.~Sun, Y.~Huang, and L.~Li, ``Embedding anti-counterfeiting features in metallic components via multiple material additive manufacturing,'' \emph{Additive Manufacturing}, vol.~24, pp. 1--12, 2018.

\bibitem{chen2019embedding}
F.~Chen, Y.~Luo, N.~G. Tsoutsos, M.~Maniatakos, K.~Shahin, and N.~Gupta, ``Embedding tracking codes in additive manufactured parts for product authentication,'' \emph{Advanced Engineering Materials}, vol.~21, no.~4, p. 1800495, 2019.

\bibitem{delmotte2019blind}
A.~Delmotte, K.~Tanaka, H.~Kubo, T.~Funatomi, and Y.~Mukaigawa, ``Blind watermarking for 3-d printed objects by locally modifying layer thickness,'' \emph{IEEE Transactions on Multimedia}, vol.~22, no.~11, pp. 2780--2791, 2019.

\bibitem{suzuki2017embedding}
M.~Suzuki, P.~Dechrueng, S.~Techavichian, P.~Silapasuphakornwong, H.~Torii, and K.~Uehira, ``Embedding information into objects fabricated with 3-d printers by forming fine cavities inside them,'' \emph{Electronic Imaging}, vol.~29, pp. 6--9, 2017.

\bibitem{li2018printracker}
Z.~Li, A.~S. Rathore, C.~Song, S.~Wei, Y.~Wang, and W.~Xu, ``Printracker: Fingerprinting 3d printers using commodity scanners,'' in \emph{Proceedings of the 2018 ACM sigsac conference on computer and communications security}, 2018, pp. 1306--1323.

\bibitem{harrison2012acoustic}
C.~Harrison, R.~Xiao, and S.~Hudson, ``Acoustic barcodes: passive, durable and inexpensive notched identification tags,'' in \emph{Proceedings of the 25th annual ACM symposium on User interface software and technology}, 2012, pp. 563--568.

\bibitem{sima2021coding}
J.~Sima, N.~Raviv, and J.~Bruck, ``On coding over sliced information,'' \emph{IEEE Transactions on Information Theory}, vol.~67, no.~5, pp. 2793--2807, 2021.

\bibitem{sima2024robust}
------, ``Robust indexing for the sliced channel: Almost optimal codes for substitutions and deletions,'' \emph{IEEE Transactions on Information Theory}, 2024.

\bibitem{lenz2019coding}
A.~Lenz, P.~H. Siegel, A.~Wachter-Zeh, and E.~Yaakobi, ``Coding over sets for dna storage,'' \emph{IEEE Transactions on Information Theory}, vol.~66, no.~4, pp. 2331--2351, 2019.

\bibitem{shomorony2021torn}
I.~Shomorony and A.~Vahid, ``Torn-paper coding,'' \emph{IEEE Transactions on Information Theory}, vol.~67, no.~12, pp. 7904--7913, 2021.

\bibitem{shomorony2020communicating}
------, ``Communicating over the torn-paper channel,'' in \emph{GLOBECOM 2020-2020 IEEE Global Communications Conference}.\hskip 1em plus 0.5em minus 0.4em\relax IEEE, 2020, pp. 1--6.

\bibitem{ravi2021capacity}
A.~N. Ravi, A.~Vahid, and I.~Shomorony, ``Capacity of the torn paper channel with lost pieces,'' in \emph{2021 IEEE International Symposium on Information Theory (ISIT)}.\hskip 1em plus 0.5em minus 0.4em\relax IEEE, 2021, pp. 1937--1942.

\bibitem{bar2023adversarial}
D.~Bar-Lev, S.~M.~E. Yaakobi, and Y.~Yehezkeally, ``Adversarial torn-paper codes,'' \emph{IEEE Transactions on Information Theory}, 2023.

\bibitem{goldstein2004minimum}
A.~Goldstein, P.~Kolman, and J.~Zheng, ``Minimum common string partition problem: Hardness and approximations,'' in \emph{International Symposium on Algorithms and Computation}.\hskip 1em plus 0.5em minus 0.4em\relax Springer, 2004, pp. 484--495.

\bibitem{jiang2012minimum}
H.~Jiang, B.~Zhu, D.~Zhu, and H.~Zhu, ``Minimum common string partition revisited,'' \emph{Journal of Combinatorial Optimization}, vol.~23, pp. 519--527, 2012.

\bibitem{damaschke2008minimum}
P.~Damaschke, ``Minimum common string partition parameterized,'' in \emph{International Workshop on Algorithms in Bioinformatics}.\hskip 1em plus 0.5em minus 0.4em\relax Springer, 2008, pp. 87--98.

\bibitem{chrobak2004greedy}
M.~Chrobak, P.~Kolman, and J.~Sgall, ``The greedy algorithm for the minimum common string partition problem,'' in \emph{International Workshop on Randomization and Approximation Techniques in Computer Science}.\hskip 1em plus 0.5em minus 0.4em\relax Springer, 2004, pp. 84--95.

\bibitem{wang2025secure}
C.~Wang, J.~Wang, M.~Zhou, V.~Pham, S.~Hao, C.~Zhou, N.~Zhang, and N.~Raviv, ``Secure information embedding in forensic 3d fingerprinting,'' in \emph{34th USENIX Security Symposium (USENIX Security 25)}, 2025, pp. 1887--1906.

\bibitem{levenshtein1970maximum}
V.~I. Levenshtein, ``Maximum number of words in codes without overlaps,'' \emph{Problemy Peredachi Informatsii}, vol.~6, no.~4, pp. 88--90, 1970.

\bibitem{gilbert1960synchronization}
E.~Gilbert, ``Synchronization of binary messages,'' \emph{IRE Transactions on Information Theory}, vol.~6, no.~4, pp. 470--477, 1960.

\bibitem{bajic2004distributed}
D.~Bajic and J.~Stojanovic, ``Distributed sequences and search process,'' in \emph{2004 IEEE International Conference on Communications (IEEE Cat. No. 04CH37577)}, vol.~1.\hskip 1em plus 0.5em minus 0.4em\relax IEEE, 2004, pp. 514--518.

\bibitem{bajic2014simple}
D.~Bajic and T.~Loncar-Turukalo, ``A simple suboptimal construction of cross-bifix-free codes,'' \emph{Cryptography and Communications}, vol.~6, no.~1, pp. 27--37, 2014.

\bibitem{chee2013cross}
Y.~M. Chee, H.~M. Kiah, P.~Purkayastha, and C.~Wang, ``Cross-bifix-free codes within a constant factor of optimality,'' \emph{IEEE Transactions on Information Theory}, vol.~59, no.~7, pp. 4668--4674, 2013.

\bibitem{bilotta2012new}
S.~Bilotta, E.~Pergola, and R.~Pinzani, ``A new approach to cross-bifix-free sets,'' \emph{IEEE Transactions on Information Theory}, vol.~58, no.~6, pp. 4058--4063, 2012.

\bibitem{blackburn2015non}
S.~R. Blackburn, ``Non-overlapping codes,'' \emph{IEEE Transactions on Information Theory}, vol.~61, no.~9, pp. 4890--4894, 2015.

\bibitem{wang2022q}
G.~Wang and Q.~Wang, ``Q-ary non-overlapping codes: A generating function approach,'' \emph{IEEE Transactions on Information Theory}, vol.~68, no.~8, pp. 5154--5164, 2022.

\bibitem{tabatabaei2015rewritable}
S.~H. Tabatabaei~Yazdi, Y.~Yuan, J.~Ma, H.~Zhao, and O.~Milenkovic, ``A rewritable, random-access dna-based storage system,'' \emph{Scientific reports}, vol.~5, no.~1, p. 14138, 2015.

\bibitem{yazdi2017portable}
S.~H.~T. Yazdi, R.~Gabrys, and O.~Milenkovic, ``Portable and error-free dna-based data storage,'' \emph{Scientific reports}, vol.~7, no.~1, p. 5011, 2017.

\bibitem{levy2018mutually}
M.~Levy and E.~Yaakobi, ``Mutually uncorrelated codes for dna storage,'' \emph{IEEE Transactions on Information Theory}, vol.~65, no.~6, pp. 3671--3691, 2018.

\bibitem{wikipedia_spherepacking}
\BIBentryALTinterwordspacing
{Wikipedia contributors}, ``Hamming bound,'' Wikipedia, The Free Encyclopedia, 2025, accessed: July 17, 2025. [Online]. Available: \url{https://en.wikipedia.org/wiki/Hamming_bound}
\BIBentrySTDinterwordspacing

\bibitem{cheng2018deterministic}
K.~Cheng, Z.~Jin, X.~Li, and K.~Wu, ``Deterministic document exchange protocols, and almost optimal binary codes for edit errors,'' in \emph{2018 IEEE 59th Annual Symposium on Foundations of Computer Science (FOCS)}.\hskip 1em plus 0.5em minus 0.4em\relax IEEE, 2018, pp. 200--211.

\end{thebibliography}
